\documentclass[12pt,titlepage]{article} 

\usepackage{soul}
\usepackage{amsthm,amsmath,amssymb,epsfig}
\usepackage[mathscr]{eucal}
\usepackage{color}
\usepackage{soul}
\usepackage[table,xcdraw]{xcolor}
\usepackage{comment}
\usepackage{todonotes}
\usepackage{verbatim}
\usepackage{listings}
\usepackage[table]{xcolor}
\newtheorem{The}{Theorem}
\newtheorem{Cor}[The]{Corollary}
\newtheorem{Lem}[The]{Lemma}

\newtheorem{Def}[The]{Definition}
\theoremstyle{definition}
\newtheorem{Rem}[The]{Remark}
\newtheorem{Exa}[The]{Example}
\theoremstyle{remark}

\numberwithin{equation}{section}
\numberwithin{The}{section}
\def\IN{\mathbb{N}}

\def\IR{\mathbb{R}}

\begin{document}

\title{A \textcolor{blue}{variance-based} importance index for systems with dependent components}
\author{Antonio Arriaza\footnote{\textcolor{blue}{Correspondence to: Antonio Arriaza, Facultad
		de Ciencias, Universidad de C\'adiz, 11510, Puerto Real, C\'adiz, Spain. E-mail: \textit{antoniojesus.arriaza@uca.es}. Telephone number: $+$34 956012775.}}$^{,a}$, Jorge Navarro$^{b}$, Miguel \'Angel Sordo$^{a}$,\\ Alfonso Su\'arez-Llorens$^{a}.$\\ 
\small{$^{a}$Departamento de Estad\'{i}stica e Investigaci\'{o}n	Operativa, Universidad de C\'{a}diz, Spain.}\\
\small{$^{b}$Departamento de Estad\'{i}stica e Investigaci\'{o}n	Operativa, Universidad de Murcia, Spain}}

\maketitle

\begin{abstract}
This paper proposes \textcolor{blue}{a variance-based} measure of importance for coherent systems with dependent and heterogeneous components. The particular cases of independent components and homogeneous components are also considered. We model the dependence structure among the components by the concept of copula. The proposed measure allows us to provide the best estimation of the system lifetime, in terms of the mean squared error, under the assumption that the lifetime of one of its components is known. We include theoretical results that are useful to calculate a closed-form of our measure and to compare two components of a system. We also provide some procedures to approximate the importance measure by Monte Carlo simulation methods. Finally, we illustrate the main results with several examples.

\vfil
\end{abstract}

\section{Introduction}
\label{S1}

A typical problem in Reliability Theory is to find the reliability function associated with a coherent system. This function is used to compute  expected lifetimes, warranty policies, residual lifetimes and other characteristics. Closely related to this problem, we find in the literature different approaches to improve the reliability of a system, such as redundancy of components or maintenance policies. However, the reliability function is insufficient to discern which component should be redundant or repaired in these improvement strategies. Let us consider  two coherent systems, with four independent and identically distributed \textcolor{blue}{(i.i.d.) components, with lifetimes $T_1 = \max(\min(X_1,X_2),\min(X_2,X_3),\min(X_3,X_4))$ and
$T_2 = \max(\min(X_1,\max(X_2,X_3,X_4)),\min(X_2,X_3,X_4)),$} \noindent where $X_i$ represents the lifetime of the $i$th component. \textcolor{blue}{From Table 2.1 in \cite{N21} and the representation of the system reliability function in terms of Samaniego's signature, see \cite{Samaniego85}, it is straightforward to prove that both systems have the same reliability function \textcolor{blue}{when the components' lifetimes are i.i.d.}} However, the best redundancy allocation (respectively, the optimal maintenance policy) does not always coincide for both systems \textcolor{blue}{since they have different structures.} 

Consequently, importance measures of components rise as a relevant tool to deal with this kind of problem. The aim is to rank a list of components (or groups of components) in terms of their relevance concerning the structure, reliability or lifetime of the system. Hence, this list can be used to design maintenance policies  to prevent early failures of the system. Furthermore, in the event of system failure, operators can follow the checklist to find the most probable cause of failure. 

According to Birnbaum \cite{Birn1969}, importance measures can be classified into three main groups. Firstly, the \textit{structure importance measure class} takes into account only the position of the components in the coherent system. Therefore, these measures are only based on the structure function of the system. This class does not involve the reliability or lifetime of the corresponding component. These types of measures are helpful in the early stages of system design, when the quality of the components are not established yet. Secondly, the \textit{reliability importance measure class}, besides the position of the components, also includes the information of the components' reliability at a fixed time point. Usually, the time considered is known as the mission time of the system. These measures analyze the change in the reliability of the system with respect to the change in the reliability of a specific component. The measures in this class are relevant for disposable systems or systems which must function during a specific period. In any case, if we consider the mission time as a time parameter, any reliability importance measure can be interpreted as a time-dependent importance measure.  Finally, the \textit{lifetime importance measure class} considers the structure of the system and the lifetime distribution of each component. This kind of measures are helpful when the functioning time of the system is indeterminate (frequently when the system has long-term service) and contains two subclasses. On the one hand, the class of time-dependent lifetime importance measures produce more detailed information depending on the time considered. Hence, the corresponding ranking of components' importance may vary over time. On the other hand,  time-independent lifetime importance measures generate more global information about the components' importance.   

Several importance measures have been proposed with different motivations for each class mentioned above. Some examples are Natvig \cite{Natv79}, Bergman \cite{Bergman85}, Norros \cite{Norros86}, Boland and El-Neweihi \cite{Boland95}, Borgonovo \cite{Borgonovo2001}, Kuo and Zuo \cite{KuoZuo2003}, Song et al. \cite{Song2012} and Navarro et al. \cite{NFFA2020}. Specially relevant are the importance measures given by Birnbaum \cite{Birn1969} and Barlow and Proschan \cite{BP75}. Over time, these measures have been generalized in order to deal with more complex systems. 

The Birnbaum measure, usually denoted by B-measure, has different definitions depending on the considered class (structure, reliability or lifetime); see, for example, \cite{Birn1969,BPbook75,Lamb} and \cite{Natv79}. The original definition of the B-lifetime measure only  considers coherent systems with independent components. Therefore, some authors have focused on giving extensions of this measure for the case of systems with dependent components, see \cite{Miziula19} and \cite{Zhang17}. One of the main characteristics of B-importance measures is that they can be used to define many other importance measures. An excellent review of the relationships between the B-measures and other importance measures can be found in \cite{Kuo12}. 

On the other hand,  Barlow and Proschan measures, denoted by BP-measures, can be defined in terms of B-measures from the corresponding structure and lifetime versions, see \cite{Kuo12}. Originally, the BP-lifetime importance was defined for the case of independent components. Iyer \cite{Iyer92} generalized the BP-lifetime importance when the components' lifetimes are jointly absolutely continuous but not necessarily independent.  Twenty years later, Marichal and Mathonet \cite{Marichal12} provided a more general extension, in terms of discrete derivatives of the structure function, with the only assumption that the joint distribution of component lifetimes has no ties. 

Some other groups of importance measures can be classified as measures of importance based on the path sets or cut sets (or both) of a coherent system. For the definitions of path and cut sets, see \cite{BPbook75}. Some examples of this type of measures are found in Fussell and Vesely \cite{Fussell72}, Butler \cite{Butler1979} and Hwang \cite{Hwang2001}. A relevant case of this kind of measures was provided by Boland et al. \cite{Boland89}. They defined the permutation importance order which was reformulated later by Meng \cite{Meng1994} and Koutras et al. \cite{Koutras1994} in terms of minimal cut set and minimal path set of a coherent system. This order provides a partial ranking of the components. Specifically, it satisfies the transitivity property and it is preserved if we replace the original structure function by the corresponding dual structure function. Given $\Phi$ the structure function of a coherent system, the associated dual coherent system  has the following structure function $\Phi^{D}(\mathbf{x})=1-\Phi(\mathbf{1}-\mathbf{x})$. Some importance measures provide, for each component, the same value for both the original system and its dual. This is the case, for example, of the B-structure measure, B-reliability measure and the BP-structure measure. Other measures, such as Fussel and Vesely \cite{Fussell72}, satisfy that the version based on path sets applied to the original system coincides with the cut sets version for the corresponding dual system for each component.

\textcolor{blue}{In this paper, we use a variance-based approach to assess the relevance of the components in a coherent system.  Based on the variance decomposition formula (or law of total variance), we use the classical coefficient of determination to measure the importance of a given component in terms of the  variance of the system that is due to that  component. Importance assessment methods based on variance decomposition have been studied extensively in many different contexts. To mention only two examples, this method has already been employed in uncertainty theory to make global sensitivity analysis of model outputs (see Borgonovo \cite{Borgonovo2007} and Iooss and Lema\^itre \cite{Ioos2014}, for a review) and in dependence theory, to study directional dependence in a copula setting (see Sungur \cite{Sungur} and Shih and Emura \cite{ShihEmura2021}). However, to the best of our knowledge, this approach is novel to the literature on reliability of coherent systems, where the system depends on the components through a structure function that needs to be studied in detail. As we will show below, the coefficient of determination offers several advantages with respect to other measures of importance in this framework.} On the one side, most of the importance measures proposed in the literature do not inform about the lifetime of the system. In the case of the proposed measure, we provide the best estimation of the system lifetime, in terms of the mean squared error, under the assumption that the lifetime of one of its components is known. This importance measure belongs to the lifetime importance measures class. Furthermore, this measure can be computed for coherent systems with both independent or dependent components. We will model that dependence by using the concept of copula function, see \cite{DS16, mull0, Ne06} and \cite{S03} for definitions and main results concerning to copula theory. Applications to coherent systems can be seen, for example, in \cite{N21, Navarroetal2021} and the references therein. On the other side, in order to be more informative, many importance measures require complex computations or estimators. We will see along the paper that our measure can be approximated by using well known estimators which produce accurate results. 

The rest of the present paper is organized as follows. In Section 2, some basic definitions and previous results are given. In Section 3, we introduce the intuitive idea in which our measure is based on. We also show how to calculate the regression curves by several examples and theoretical results. Some conditions which imply the ordering of the components, in terms of the proposed measure, are given in Section 4. In Section 5, we provide some procedures to \textcolor{blue}{approximate} the measure of importance, by using Monte Carlo methods. \textcolor{blue}{Furthermore, we illustrate the applicability of the importance measure by studying a system in the context of naval engineering.} Finally, we place the conclusions in Section 6.  

\section{Notation and preliminary results}

Let us consider a coherent system with lifetime $T$ and component lifetimes $X_1,\dots,X_n$.  Then, it is well known (see e.g., \cite{BPbook75}, p.\ 12) that 
$$T=\phi(X_1,\dots, X_n)=\max_{i=1,\dots,r} \min_{j\in P_i} X_j,$$
where $P_1,\dots, P_r\subseteq \{1,\dots,n\}$ are the minimal path sets of the system. A set $P\subseteq\{1,\dots,n\}$ is a {\it path set} if the system works when all the components in the set $P$ work. A path set is a {\it minimal path set} if does not contain other path sets. 

We will use the following notation for the lifetime of the  series system with components in a set $P$,  $X_P:=\min_{j\in P} X_j$. Then, its reliability (or survival)  function is
\begin{equation}
\label{relseriesys}
\bar F_P(t):=\Pr(X_P>t)=\Pr(\cap_{j\in P} \{X_j>t\})    
\end{equation}
and the system reliability function $\bar F_T(t):=\Pr(T>t)$ can be obtained as
\begin{equation}
\label{relsys0}
\bar F_T(t)=\sum_{i=1}^r \bar F_{P_i}(t) - \sum_{i=1}^{r-1} \sum_{j=i+1}^{r}  \bar F_{P_i\cup P_j}(t)+\dots (-1)^{r+1}\bar F_{P_1\cup \dots \cup P_r}(t)
\end{equation}
for all $t\geq 0$ (see, e.g., \cite{N21}, p.\ 37). This expression is called {\it minimal path set representation}.

We assume that the components can be dependent and that this dependence structure is represented by a {\it survival copula} $\widehat C$ which can be used to represent their joint reliability function as 
\begin{equation}\label{SC} 
\Pr(X_1>x_1,\dots, X_n>x_n)=\widehat C( \bar F_1(x_1),\dots, \bar F_n(x_n))
\end{equation}
(see e.g. \cite{Ne06}, p. 32, or \cite{DS16}, p. 33), where $\bar F_i(t)=\Pr(X_i>t)$ is the reliability function of the $i$th component  for $i=1,\dots,n$. The independence case is represented by the product copula $\widehat C(u_1,\dots,u_n)=u_1 \cdots u_n$.

We shall use the following additional notation. 
For $\mathbf{u}=(u_1,\dots,u_n)$ and $P\subseteq\{1,\dots,n\}$, $\mathbf{u}_P:=(u^P_1,\dots,u^P_n)$ with $u^P_i=u_i$ for $i\in P$ and $u^P_i=1$ for $i\notin P$. For example, if $\mathbf{u}=(u_1,u_2,u_3)$, then 
$\mathbf{u}_{\{1,3\}}=(u_1,1,u_3)$. Then, for a given $n$-dimensional copula $C$,  we define 
$C_P(\mathbf{u}):=C(\mathbf{u}_P)$. Note that this is the copula for the marginal distribution of the random vector formed by the variables with the indices included in $P$.  Thus, in the case of dependent components, the reliability function in \eqref{relseriesys} can we rewritten as 
$$\bar F_P(t)= \widehat C_P (\bar F_1(t),\dots,\bar F_n(t))$$
for all $t\geq 0$. By using this expression in the minimal path set representation obtained in \eqref{relsys0}, we can write the system reliability function as
$$\bar F_T(t)=\bar Q (\bar F_1(t),\dots,\bar F_n(t)),
$$
that is, as a distortion of the component reliability functions, where the distortion function 
$$\bar Q (\mathbf{u})= \sum_{i=1}^r \widehat C_{P_i}(\mathbf{u}) - \sum_{i=1}^{r-1} \sum_{j=i+1}^{r}  \widehat C_{P_i\cup P_j}(\mathbf{u})+\dots +(-1)^{r+1}\widehat C_{P_1\cup \dots \cup P_r}(\mathbf{u})$$
only depends on the system structure (the minimal path sets) and the dependence structure (the survival copula).

\textcolor{blue}{Analogously,  we define the function $\hat{C}_{i,P}$ as  
$$\hat{C}_{i,P}(\mathbf{u})=\partial_i \hat{C}(\mathbf{u}_{P\cup \{i\}}),$$
where $\partial_i \hat{C}$ is  the partial derivative of $\hat{C}$ with respect to its $i$th variable. Similarly, $\partial_{i,j} \hat{C}:=\partial_i \partial_j \hat{C}$ and so on. It is well known that the partial derivatives of a copula are related with conditional distributions, see, for example, \cite{DS16} (p.\ 91), \cite{Ne06}, (p.\ 217), \cite{S03} (pp.\ 16-22) and \cite{Ne99}, (p.\ 175).} 

The following lemmas will be useful to compute the expectation of a system when we know the failure time of a component.  

\begin{Lem}\label{l1}
Let us assume that $(X_1,\dots,X_n)$ has a joint absolutely continuous distribution and survival copula $\widehat C$. Let $P\subseteq\{1,\dots,n\}$, $X_P=\min_{j\in P}X_j$ and $i\in\{1,\dots,n\}$. If $i\in P$, $P\neq \{i\}$,  $f_i(x)=-\bar F'_i(x)>0$ and 
\begin{equation}\label{c1}
\lim_{u_j\to 0^+} \partial_{i}\widehat C(u_1,\dots,u_n)=0 
\end{equation} 
holds for all $j\in P-\{i\}$, then
$$\Pr(X_P>t|X_i=x)=\left\{\begin{array}{ccc}
\widehat{C}_{i,P}(\bar F_1(y_1),\dots,\bar F_n(y_n) )& &0\leq t< x,\\
0											&& x\leq t,\\
\end{array}\right.
$$
where $y_i=x$ and  $y_j=t$  for $j\neq i$. Moreover,
$$E(X_P|X_i=x)=\int_0^x \widehat C_{i,P}(\bar F_1(y_1),\dots,\bar F_n(y_n) ) dt.$$
\end{Lem}
\begin{proof}
Clearly, $\Pr(X_P>t|X_i=x)=0$ if $t\geq x$ since $X_P\leq x$ when $X_i=x$ and $i\in P$. For $0\leq t\leq x$, we have
$$\Pr(X_P>t|X_i=x)=\Pr(Y_i>t|X_i=x),$$
where $Y_i:=X_{P-\{i\}}=\min_{j\in P-\{i\}}X_j$. The absolutely continuous joint reliability function $\bar G$ of $(X_i, Y_i)$ is 
$$\bar G(x,y):=\Pr(X_i>x,Y_i>y)=\widehat C(\bar F_1(z_1),\dots,\bar F_n(z_n)),$$
where $z_i=x$,
$z_j=y$ if $j\in P$, and $z_j=-\infty$ if $j\notin P\cup \{i\}$. 
 Hence
$$\partial_1 \bar G(x,y)=-f_i(x)\widehat C_{i,P}(\bar F_1(z_1),\dots,\bar F_n(z_n)).$$
The joint probability density function  (pdf) $g$ of $(X_i,Y_i)$ is 
$g(x,y)= \partial_{1,2} \bar G(x,y)$
and the pdf of $Y_i|X_i=x$ is
$$g_{2|1}(y|x)=\frac{ \partial_{1,2} \bar G(x,y)}{f_i(x)}$$
for $y\geq 0$ (zero elsewhere). Then its reliability function $\bar G_{2|1}(t|x):=\Pr(Y_i>t|X_i=x)$ is
$$\bar G_{2|1}(t|x)=\int_t^\infty \frac{ \partial_{1,2} \bar G(x,y)}{f_i(x)}dy=\left[\frac{ \partial_{1} \bar G(x,y)}{f_i(x)}\right]_{y=t}^\infty=-\frac{ \partial_{1} \bar G(x,t)}{f_i(x)}$$
since \eqref{c1} 
holds for $j\in P-\{i\}$. Finally, we use the above expression for $\partial_{1} \bar G$ to get the stated result when $f_{i}(x)> 0$.
\end{proof}

\quad

Note that the distribution of $\{X_P|X_i=x\}$ is a mixture of an absolutely continuous distribution and  a discrete atom distribution with mass at $x$. Also note that if $P=\{i\}$, then $\Pr(X_P>t|X_i=x)=1$ for $0\leq t<x$ and $\Pr(X_P>t|X_i=x)=0$ for $t\geq x$. Moreover, in this case $E(X_P|X_i=x)=x$ for all $x$ such that $f_i(x)>0$.

Next we state the result for $i\notin P$. This result   could also be obtained from the analogous result of Theorem 3.4.1 in \cite{DS16} or  Corollary 2.24 in \cite{S03} for the respective survival functions.

\begin{Lem}\label{l2}
	Let us assume that $(X_1,\dots,X_n)$ has a joint absolutely continuous distribution. Let $P\subseteq\{1,\dots,n\}$ and $i\in\{1,\dots,n\}$. If $i\notin P$, $f_i(x)>0$ and \eqref{c1} holds for $j\in P$, then
	$$\Pr(X_P>t|X_i=x)=\widehat C_{i,P}(\bar F_1(y_1),\dots,\bar F_n(y_n) )$$
	for $ t\geq 0$, where $y_i=x$ and  $y_j=t$ for $j\neq i$. Moreover
	$$E(X_P|X_i=x)=\int_0^\infty \widehat C_{i,P}(\bar F_1(y_1),\dots,\bar F_n(y_n) )dt.$$
\end{Lem}
\begin{proof}
	The absolutely continuous joint reliability function $\bar G$ of $(X_i, X_P)$ is 
	$$\bar G(x,y):=\Pr(X_i>x,X_P>y)=\widehat C(\bar F_1(z_1),\dots,\bar F_n(z_n)),$$
	where $z_i=x$, $z_j=y$ if $j\in P$, and $z_j=-\infty$ if $j\notin P\cup \{i\}$. Hence
	$$\partial_1 \bar G(x,y)=-f_i(x)\widehat C_{i,P}(\bar F_1(z_1),\dots,\bar F_n(z_n)).$$
	The joint pdf $g$ of $(X_i,X_P)$ is 
	$g(x,y)= \partial_{1,2} \bar G(x,y)$
	and the conditional pdf of $X_P|X_i=x$ is
	$$g_{2|1}(y|x)=\frac{ \partial_{1,2} \bar G(x,y)}{f_i(x)}$$
	for $ y\geq 0$. Then 
	the reliability function $\bar G_{2|1}(t|x):=\Pr(X_P>t|X_i=x)$ can be obtained as
	$$\bar G_{2|1}(t|x)=\int_t^\infty \frac{ \partial_{1,2} \bar G(x,y)}{f_i(x)}dy=\left[\frac{ \partial_{1} \bar G(x,y)}{f_i(x)}\right]_{y=t}^\infty=-\frac{ \partial_{1} \bar G(x,t)}{f_i(x)}$$
for $t\geq 0$,	whenever \eqref{c1} holds for $j\in P$. Finally, we use the above expression for $\partial_{1} \bar G$ to get the stated result.
\end{proof}

\quad

These two lemmas can be used jointly with the minimal path set representation to get the following representation for the conditional distribution of a system when we know the failure time of a component. This result is of independent interest and can be used to predict the system failure time from a component failure time.

\begin{The}\label{th1}
Let us assume that the component lifetimes of a system $(X_1,\dots,X_n)$ have a joint absolutely continuous distribution and that  $f_k(x)>0$  for an $x\geq 0$ and  a $k\in \{1,\dots,n\}$.  Let $T$ be the system lifetime and let $P_1,\dots, P_r$ be its minimal path sets. Then
\begin{align*}
	\Pr(T>t|X_k=x)&=\sum_{i=1}^r \Pr(X_{P_i}>t|X_k=x) - \sum_{i=1}^{r-1} \sum_{j=i+1}^{r}  \Pr(X_{P_i\cup P_j}>t|X_k=x)\\
	& \quad +\dots (-1)^{r+1}\Pr(X_{P_1\cup \dots \cup P_r}>t|X_k=x)
\end{align*}
for all $t\geq 0$.  
\end{The}
\begin{proof}
As $T=\max_{i=1,\dots,r} X_{P_i}$, then
$$\Pr(T>t|X_i=x) =
\Pr(\max_{i=1,\dots,r} X_{P_i}>t|X_i=x)=
\Pr\left(\cup_{i=1}^r \{ X_{P_i}>t\} |X_i=x\right)
$$
and by applying the inclusion-exclusion formula we get the stated result.
\end{proof}

\quad

Note that we obtain a generalized mixture of reliability functions that can be computed from one of  the two preceding lemmas (depending if $k$ is included in $P$ or not). 

The joint reliability function $\bar G$ of $T$ and $X_k$ can be obtained in a similar way.  Note that it can be stated as a generalized distortion based on $\bar F$, see \cite{NCLD21}. Also note that the distribution of  $(T,X_k)$ might  have a singular part.  This function $\bar G$ can be used to compute the covariance between $T$ and $X_k$ as
$Cov(T,X_k)=E(TX_k) - E(T)E(X_k)$, where
$E(TX_k)=\int_0^\infty\int_0^\infty \bar G(x,y) dx dy$.
Some examples are given below.

\section{\textcolor{blue}{ A variance-based measure of importance}}
\label{S2}

Let us consider a coherent system with lifetime
$T=\phi(X_1,\dots,X_n)$ based on possibly dependent components
with lifetimes $X_1,\dots,X_n$, where $\phi$ is the structure
function of the system (see \cite{BP81}, Chapter 1). Let us assume that components' lifetimes are non necessarily identically distributed having reliability functions $\bar{F}_i(t)=\mbox{Pr} (X_i >t)$, $i=1,\ldots, n$. 

Given the $i$th and $j$th components, we will first consider the bivariate random vectors $(X_i,T)$ and $(X_j,T)$ and the conditional random variables $\{T|X_i = x \}$ and $\{T|X_j = x\}$, for $x \geq 0$. Then, we define the classical regression curves
\begin{equation*}
\label{eq2.1}
m_i(x)=E\big[T\big|X_i = x \big] \quad\text{and} \quad m_j(x)=E\big[T\big|X_j = x \big],\quad x \geq 0,
\end{equation*}
and the error curves
\begin{equation*}
\label{eq2.2}
e_i(x)=Var\big[T\big|X_i = x \big] \quad\text{and} \quad e_j(x)=Var\big[T\big|X_j = x \big],\quad x \geq 0,
\end{equation*}
provided that they exist. It is well known that \(m_i(X_i)\),  \(e_i(X_i)\) and \(m_j(X_j)\), \(e_j(X_j)\) are pairs of univariate random variables that satisfy the ``law of total variance'', i.e.,
\begin{equation}
\label{eq2.3}
Var(T) = Var(m_i(X_i))+E(e_i(X_i)) = Var(m_j(X_j))+E(e_j(X_j)).
\end{equation}
Note that \(m_i(X_i)\) and \(e_i(X_i)\) are also denoted by $E[T\big|X_i]$ and $Var[T\big|X_i]$ in the literature, respectively. 
Consider now \(m_i(X_i)\) and \(m_j(X_j)\). Intuitively, if \(X_i\) does not have much influence on \(T\) then, by observing \(X_i=x\) we learn ``almost nothing" about \(T\); that is, \(T\) does not vary much with \(X_i\), or, in other words, ``\(T\) does not inherit much of the variability of \(X_i\)''. As a result, \(m_i(X_i)\) has a small variability. On the other hand, intuitively, if \(X_j\) has a strong influence on  \(T\) then, by observing \(X_j=x\) we learn ``a lot" about \(T\); that is, \(T\) varies much with \(X_j\), or, in other words, ``\(T\) inherits much of the variability of \(X_j\)''. As a result, \(m_j(X_j)\) has a large variability. In conclusion, if
\begin{equation}
\label{eq2.5}
\textup{Var}(m_i(X_i))<\textup{Var}(m_j(X_j)),
\end{equation}
then we have an indication that the component \(X_i\) has less influence on $T$ than \(X_j\). 

Following a similar argument, just observing \eqref{eq2.3} and \eqref{eq2.5} we deduce that if
\[
E(e_i(X_i))>E(e_j(X_j)),
\]
then, again, we have an indication that the component \(X_i\) has less influence on $T$ than \(X_j\). 

From the above intuitive discussion, we propose the classical  coefficient of determination to measure the importance of the $i$th component. We will define $R_i^2$ as the proportion of the variance in the lifetime $T$ of the system that is predictable from the $i$th component. 

\begin{Def} 
Given a coherent system with $n$ components, we define the regression importance index of the $i$th component as
\begin{equation}\label{rii}
R_i^2 = \frac{Var(m_i(X_i))}{Var(T)}=1- \frac{E(e_i(X_i))}{Var(T)}, \, \, i=1,\ldots,n.    
\end{equation}
\end{Def} 
\textcolor{blue}{The coefficient of determination $R_i^2 $ is a fundamental tool in quantitative sensitivity analysis, regression analysis, statistical dependence and other statistical models. In sensitivity analysis, the point of departure is a mathematical model $Y=f(X_1,...,X_n)$ where some of the input factors are uncertain  and the objective is to rank them in order of importance (see, for example  the book by Saltelli et al. \cite{Saltellietal}). In this context, $R_i^2 $ (with $Y$ replaced by $T$) is called the importance measure or sensitivity index. When the inputs are independent, $Var(m_i(X_i))$ can be interpreted in the context of a general variance decomposition scheme proposed by Sobol \cite{Sobol1990} and the index $R_i^2 $ is known as the Sobol's index. Sungur \cite{Sungur} considered the transformed pair $(U,V)=(F_i(X_i),F_Y(Y))$ and used the index $Var(E[V|U])/Var(V)$  to study directional dependence in a copula framework. In this context, the index is known as the copula correlation ratio (see Shih and Emura, \cite{ShihEmura2021}). However, a major difference between Sungur's approach and ours is that $R_i^2$ in \eqref{rii} is sensitive to the marginal distributions, whereas Sungur's index only depends on the copula. It is worth mentioning that whereas variance-based sensitivity measures are generally estimated numerically, in our context the index $R_i^2 $ can sometimes be computed analytically via the representation given in Theorem \ref{th1} for the conditional distribution of a system when we know the failure time of a component. We show some examples in Section \ref{compute}. }

\subsection{How to compute the regression curves.}\label{compute}

As we have mentioned, the components' lifetimes $X_1,\dots,X_n$ can be dependent, and this dependence will be represented by the corresponding survival copula $\widehat{C}$. Hence, the joint survival function can be expressed as in \eqref{SC}.

Let us assume that we have a coherent system with $n$ possible dependent components and we wish to calculate the expected lifetime of the system given a value of the $i$th component. The following corollary states how to compute the corresponding regression curve. The proof is straightforward from Theorem \ref{th1} and, therefore, it is omitted.

\begin{Cor}\label{Cor1}
Let us assume that the component lifetimes of a system $(X_1,\dots,X_n)$ have a joint absolutely continuous distribution and that  $f_k(x)>0$  for an $x\geq 0$ and a $k\in \{1,\dots,n\}$.  Let $T$ be the system lifetime and let $P_1,\dots, P_r$ be its minimal path sets. Then
$$m_k(x) =\sum_{j=1}^r E\big[ X_{P_j} | X_k=x \big]-\sum_{i=1}^{r-1}\sum_{j=i+1}^{r}E\big[ X_{P_i\cup P_j} | X_k=x \big]+\dots (-1)^{r+1} E\big[ X_{P_1\cup\dots\cup P_r} | X_k=x \big].$$
\end{Cor}

Next we provide some examples to clarify the computation of $m_i(x)$ and $R^2_i$ for the $i$th component.

\begin{Exa} \label{ex1_exp}
Let us consider a series system with two dependent components and lifetime $T = \min (X_1,X_2)$, where the dependence between $X_1$ and $X_2$ is modelled by the survival copula $\widehat{C}$ and the reliability functions of both components are given by $\bar{F}_1$ and $\bar{F}_2$, respectively. Let us denote by $P=\{1,2\}$ the unique minimal path set of $T$. From Theorem \ref{th1}, we obtain 
\begin{equation}
\label{m1seriessys}
m_1(x)= E\big[X_P \big |X_1 = x \big] = \int_{0}^{x} \widehat{C}_{1,P} (\bar{F}_{1}(x), \bar{F}_{2}(t))dt.
\end{equation}
Then, it is apparent that
$$
m_1(X_1)=_{st} \int_{0}^{F_{1}^{-1}(U)} \widehat{C}_{1,P}(1-U, \bar{F}_{2}(t))dt,
$$
where $U$ is a uniform random variable in the interval $(0,1)$, denoted by $U\sim U(0,1)$. 

In the particular case of independent components $\widehat{C}(u_1,u_2)=u_1 u_2$, and $m_1(x)$ in \eqref{m1seriessys} takes the following form
$$m_1(x) = \int_{0}^{x} \bar{F}_{2}(t)dt.$$
If we also consider that both random variables are exponentially distributed, denoted by $X_i\sim Exp(\lambda_i)$ for $i=1,2$ (with $\lambda_1$ and $\lambda_2$ the failure rates of $X_1$ and $X_2$, respectively), we easily obtain 
$$ 
m_1(x)= \int_{0}^{x} \bar{F}_{2}(t)dt = \frac{1}{\lambda_2}\left ( 1-\exp\{-\lambda_2 x\} \right ).
$$
From the above expression, a straightforward computation shows that
$$
m_1(X_1) =_{st} \frac{1}{\lambda_2}\left( 1-\exp\{-\lambda_2 X_1\} \right)
$$ 
and, therefore
$$
Var(m_1(X_1)) = \frac{\lambda_1}{(\lambda_1+2\lambda_2)(\lambda_1+\lambda_2)^2}.
$$
Now, from the well-known fact that the minimum of two independent random variables with exponential distributions has also an exponential distribution, we have that  $Var(T)=1/(\lambda_1+\lambda_2)^2$. Therefore
\begin{equation} 
\label{expinflu}
R_1^2= \frac{Var(m_1(X_1))}{Var(T)} = \frac{\lambda_1}{\lambda_1+2\lambda_2}. 
\end{equation}
From the symmetry of the considered system (series system with two independent components), we obtain the importance measure for the second component as
$$
R_2^2= \frac{Var(m_2(X_2))}{Var(T)} = \frac{\lambda_2}{\lambda_2+2\lambda_1}. 
$$

\begin{Rem}
\label{obsR2seriesystem}
Expression \eqref{expinflu} has an interesting interpretation. Firstly, if $\lambda_1=\lambda_2$ both components have the same influence on the lifetime of the system, $R_1^2=R_2^2$, as expected. Secondly, if $\lambda_1<\lambda_2$, then $R_1^2 < R_2^2$, i.e., that component with a higher failure rate (weaker component) produces a major influence on the system. This result is in concordance with several results of importance measures (see, for example, Theorem 3.8 in \cite{BP75}, Section 3.4.3 in \cite{Birn1969} or Theorem 3.3 in \cite{Natv79}). Finally, if $\lambda_1$ tends to $+\infty$, then $R_1^2$ tends to $1$ and $R_2^2$ tends to $0$, that is, the lifetime of the system could be explained exclusively by the lifetime of the component $X_1$, and $X_2$ would not have influence at all on the system. Similarly, if $\lambda_1$ tends to $0$, then $R_1^2$ tends to $0$ and $R_2^2$ tends to $1$, in this is case $X_1$ never fails, and therefore, it produces no effects on the lifetime of the system.
\end{Rem}

Observe that the bivariate case can be easily extended to the general case, i.e., a series system with $n$ dependent components and lifetime $T = \min ( X_1, \ldots X_n)$, where the component lifetimes have a survival copula $\widehat{C}$ and reliability functions $\bar{F}_{i}(t)$ for $i=1,\ldots, n$. In this case
$$
m_i(x)= E\big[X_P \big |X_i = x \big] = \int_{0}^{x} \widehat{C}_{i,P} (\bar{F}_{1}(t),\ldots,\bar{F}_{i-1}(t),\bar{F}_{i}(x),\bar{F}_{i+1}(t),\ldots, \bar{F}_{n}(t))dt,
$$
where $\widehat{C}_{i,P}(\mathbf{u})=\partial_i \widehat{C}(\mathbf{u}_{P\cup \{i\}})$. Thus, in the particular case of having an independence survival copula  $\widehat{C}(u_1,\ldots,u_n)=\displaystyle \prod_{i=1}^{n}u_i$, and exponentially distributed components, $X_i\sim Exp(\lambda_i)$ for $i=1,\ldots, n$, we obtain
\begin{equation}
\label{rc_minimalPS}
m_i(x) = \frac{1}{\displaystyle \sum_{\underset{j\not = i}{j=1}}^{n} \lambda_j} \left ( 1-\exp\left(- \displaystyle \sum_{\underset{j\not = i}{j=1}}^{n} \lambda_j x \right) \right )
\end{equation}
and
$$
R_i^2= \frac{\lambda_i}{\lambda_i+2 \displaystyle \sum_{\underset{j\not = i}{j=1}}^{n} \lambda_j}
$$
for all $i=1,\ldots,n$. Observe that $R_i^2$ has a similar interpretation as that given in Remark \ref{obsR2seriesystem}. Furthermore, it is not difficult to see that the weakest component produces the highest value of the new measure (as expected for a series system). 
\end{Exa}

\begin{Rem}
	From Corollary \ref{Cor1} and the formula given in \eqref{rc_minimalPS}, it is easy to provide a closed expression of regression curves $m_i(x)$ for any system with independent and exponentially distributed components. 
\end{Rem}

\begin{Rem}
	In Example \ref{ex1_exp} we obtain a closed-form for the measures of influence. Of course, this is not always possible. However, we will be able to approximate these measures by using simulation methods as we will see in Section \ref{estimationIndex}.    
\end{Rem}

\begin{Exa}
\label{Example34}
Let us consider now the system with lifetime $T=\max(X_1,\min(X_2,X_3))$ and dependent components with survival copula $\hat{C}$. The minimal path sets are $P_1=\{1\}$ and $P_2=\{2,3\}$. The reliability function of $\{T|X_1=x\}$ is
$$\Pr(T>t|X_1=x)=\Pr(X_{P_1}>t|X_1=x)+\Pr(X_{P_2}>t|X_1=x)-\Pr(X_{P_1\cup P_2}>t|X_1=x)$$
for $t>x$ ($1$ elsewhere). Obviously, $\Pr(X_{P_1}>t|X_1=x)=0$ for $t\geq x$. From the Lemma \ref{l2}, we have
$$\Pr(X_{P_2}>t|X_1=x)=\partial_1 \hat{C}(\bar F_1(x),\bar F_2(t),\bar F_3(t))$$ 
for $t\geq 0$ and from Lemma \ref{l1},
$$\Pr(X_{P_1\cup P_2}>t|X_1=x)=\partial_1 \hat{C}(\bar F_1(x),\bar F_2(t),\bar F_3(t))$$
for $0\leq t<x$ (zero for $t\geq x$). Therefore
$$\Pr(T>t|X_1=x)=\Pr(X_{P_2}>t|X_1=x)=\partial_1 \hat{C}(\bar F_1(x),\bar F_2(t),\bar F_3(t))$$
for $t\geq x$ ($1$ elsewhere). Then
$$m_1(x)=E(T|X_1=x)=x+\int_x^\infty \partial_1 \hat{C}(\bar F_1(x),\bar F_2(t),\bar F_3(t))dt\,\,\,\mbox{for}\,\,\,x\geq 0.$$
In particular, if the components are independent, then 
$$m_1(x)=x+\int_x^\infty \bar F_2(t)\bar F_3(t)dt\,\,\,\mbox{for}\,\,\,x\geq 0.$$
If all the components are exponentially distributed, $X_i\sim Exp(\lambda_i)$ for $i=1,2,3$, then 
\begin{equation}
\label{m1Ex36}
m_1(x)=x+\int_x^\infty \exp(-(\lambda_2+\lambda_3)t)dt=x+\frac{1}{\lambda_2+\lambda_3} \exp(-(\lambda_2+\lambda_3)x)
\end{equation}
for $x\geq 0$. Note that the joint distribution of $(X_1,  T)$ has a singular part with 
\begin{align*}
\Pr(T=X_1)&=1-\Pr(X_1<\min(X_2,X_3))=1-\int_0^\infty \lambda_1\exp(-(\lambda_1+\lambda_2+\lambda_3)t)dt\\
          &=1-\frac{\lambda_1}{\lambda_1+\lambda_2+\lambda_3}=\frac{\lambda_2+\lambda_3}{\lambda_1+\lambda_2+\lambda_3}.
\end{align*}
Analogously, for the second component we obtain
$$\Pr(T>t|X_2=x)=\Pr(X_{P_1}>t|X_2=x)+\Pr(X_{P_2}>t|X_2=x)-\Pr(X_{P_1\cup P_2}>t|X_2=x)$$
for $t\geq 0$. To obtain the first reliability function in the mixture we use Lemma \ref{l2}. The second and third reliability functions are obtained by using Lemma \ref{l1}, getting that
$$\Pr(X_{P_1}>t|X_2=x)=\partial_2 \hat{C}(\bar F_1(t),\bar F_2(x),1)\,\,\,\, \mbox{for}\,\,\,\, t\geq 0.$$
$$\Pr(X_{P_2}>t|X_2=x)=\partial_2 \hat{C}(1,\bar F_2(x),\bar F_3(t))\,\,\,\, \mbox{for}\,\,\,\, 0\leq t< x\,\,(\mbox{zero for}\,\,t\geq x).$$
$$\Pr(X_{P_1\cup P_2}>t|X_2=x)=\partial_2 \hat{C}(\bar F_1(t),\bar F_2(x),\bar F_3(t))\,\,\,\, \mbox{for}\,\,\,\, 0\leq t< x\,\,(\mbox{zero for}\,\,t\geq x).$$
If the components are independent, then

$$\Pr(T>t|X_2=x)=\left\{\begin{array}{ccc}
\bar{F}_1(t)+\bar{F}_3(t)-\bar{F}_1(t)\bar{F}_3(t)& &0\leq t< x,\\
\bar{F}_1(t)							  & & x\leq t.\\
\end{array}\right.
$$
In particular, if all the components are exponentially distributed, then
$$m_2(x)=\int_0^x (\exp(-\lambda_1\,t)+\exp(-\lambda_3\,t)-\exp(-(\lambda_1+\lambda_3)\,t))dt+\int_x^\infty \exp(-\lambda_1\,t)dt,$$
that is,
\begin{equation}
\label{m2Ex36}
m_2(x)=\frac{1}{\lambda_1}+\frac{1}{\lambda_3}-\frac{1}{\lambda_1+\lambda_3}-\frac{1}{\lambda_3}\exp(-\lambda_3\,x)+\frac{1}{\lambda_1+\lambda_3}\exp(-(\lambda_1+\lambda_3)\,x).
\end{equation}
From \eqref{m1Ex36} we consider now the random variable $$Z_1:=m_1(X_1)=X_1+\frac{1}{\lambda_2+\lambda_3} \exp(-(\lambda_2+\lambda_3)X_1)$$
with variance $Var[Z_1]=E[Z_1^2]-E[Z_1]^2$ computed as follows
$$E[Z_1]=E(X_1)+\frac{1}{\lambda_2+\lambda_3}E(\exp(-(\lambda_2+\lambda_3)X_1))=\frac{1}{\lambda_1}+\frac{1}{\lambda_2+\lambda_3}\frac{\lambda_1}{\lambda_1+\lambda_2+\lambda_3},$$
\begin{align*}
E[Z_1^2]&=\int_0^\infty \lambda_1(x+\frac{1}{\lambda_2+\lambda_3} \exp(-(\lambda_2+\lambda_3)x))^2 \exp(-\lambda_1 x)dx\\
        &=\frac{2}{\lambda_1^2}+\frac{\lambda_1}{(\lambda_2+\lambda_3)^2(\lambda_1+2\lambda_2+2\lambda_3)}+\frac{2\lambda_1}{(\lambda_2+\lambda_3)(\lambda_1+\lambda_2+\lambda_3)^2}.
\end{align*}

\noindent Hence,
\begin{align*}
Var[m_1(X_1)]=&\frac{2}{\lambda_1^2}+\frac{\lambda_1}{(\lambda_2+\lambda_3)^2(\lambda_1+2\lambda_2+2\lambda_3)}+\frac{2\lambda_1}{(\lambda_2+\lambda_3)(\lambda_1+\lambda_2+\lambda_3)^2}\\
&-\left (\frac{1}{\lambda_1}+\frac{\lambda_1}{(\lambda_2+\lambda_3)(\lambda_1+\lambda_2+\lambda_3)} \right )^2.
\end{align*}

From \eqref{m2Ex36} we consider now the random variable
$$Z_2:=m_2(X_2)= \frac{1}{\lambda_1}+\frac{1}{\lambda_3}-\frac{1}{\lambda_1+\lambda_3}-\frac{1}{\lambda_3}\exp(-\lambda_3\,X_2)+\frac{1}{\lambda_1+\lambda_3}\exp(-(\lambda_1+\lambda_3)\,X_2),$$
\noindent with 
$$E[Z_2] =\frac{1}{\lambda_1}+\frac{\lambda_1}{(\lambda_2+\lambda_3)(\lambda_1+\lambda_2+\lambda_3)},$$
\begin{align*}
E[Z_2^2]=& \int_0^\infty \lambda_2m_2(x)^2 \exp(-\lambda_2 x)dx\\
        =&\,\, \kappa^2+\frac{\lambda_2}{\lambda_3^2(\lambda_2+2\lambda_3)}-\frac{2\lambda_2\kappa}{\lambda_3(\lambda_2+\lambda_3)}
           +\frac{\lambda_2}{(\lambda_1+\lambda_3)^2(2\lambda_1+\lambda_2+2\lambda_3)}\\
           &+\frac{2\lambda_2\kappa}{(\lambda_1+\lambda_3)(\lambda_1+\lambda_2+\lambda_3)}-\frac{2\lambda_2}{\lambda_3(\lambda_1+\lambda_3)(\lambda_1+\lambda_2+2\lambda_3)},
\end{align*}
\noindent where $\kappa = \frac{(\lambda_1+\lambda_3)^2-\lambda_1\lambda_3}{\lambda_1\lambda_3(\lambda_1+\lambda_3)}.$ Hence, 
\begin{align*}
Var[m_2(X_2)]=&\kappa^2+\frac{\lambda_2}{\lambda_3^2(\lambda_2+2\lambda_3)}-\frac{2\lambda_2\kappa}{\lambda_3(\lambda_2+\lambda_3)}
           +\frac{\lambda_2}{(\lambda_1+\lambda_3)^2(2\lambda_1+\lambda_2+2\lambda_3)}\\
           &+\frac{2\lambda_2\kappa}{(\lambda_1+\lambda_3)(\lambda_1+\lambda_2+\lambda_3)}-\frac{2\lambda_2}{\lambda_3(\lambda_1+\lambda_3)(\lambda_1+\lambda_2+2\lambda_3)}\\
           &-\left( \frac{1}{\lambda_1}+\frac{\lambda_1}{(\lambda_2+\lambda_3)(\lambda_1+\lambda_2+\lambda_3)}  \right)^2.
\end{align*}

In order to calculate the proposed index for the first and second component, we need to compute $Var(T)$. The reliability function of $T$ is 
$$\bar{F}_{T}(t)=\exp(-\lambda_1\,t)+\exp(-(\lambda_2+\lambda_3)\,t)-\exp(-(\lambda_1+\lambda_2+\lambda_3)\,t)$$
\noindent for $t\geq 0$. Therefore, 
\begin{equation}
\label{ET}
E(T) = \frac{1}{\lambda_1}+\frac{1}{\lambda_2+\lambda_3}-\frac{1}{\lambda_1+\lambda_2+\lambda_3}    
\end{equation}
and
\begin{equation}
\label{ET2}
E(T^2) = \frac{2}{\lambda_1^2}+\frac{2}{(\lambda_2+\lambda_3)^2}-\frac{2}{(\lambda_1+\lambda_2+\lambda_3)^2}.    
\end{equation}
From \eqref{ET} and \eqref{ET2} we get 
$$Var(T) = \frac{2}{\lambda_1^2}+\frac{2}{(\lambda_2+\lambda_3)^2}-\frac{2}{(\lambda_1+\lambda_2+\lambda_3)^2}-\left(\frac{1}{\lambda_1}+\frac{1}{\lambda_2+\lambda_3}-\frac{1}{\lambda_1+\lambda_2+\lambda_3}\right)^2.$$
 
 In the case of taking $\lambda_1=\lambda_2=\lambda_3=1$, we obtain that
\begin{equation*}
\label{R12Ex34}
R_1^2 = \frac{Var(m_1(X_1))}{Var(T)}=\frac{4/5}{11/12}=\frac{48}{55}=0.872727273,    
\end{equation*}
\begin{equation*}
\label{R22Ex34}
R_2^2 = \frac{Var(m_2(X_2))}{Var(T)}=\frac{1/45}{11/12}=\frac{12}{495}=0.02424242.    
\end{equation*}

Obviously, by the symmetry of the system and copula, we conclude that $R_3^2=R_2^2$. In the case of independent and exponentially distributed components, with the same hazard rate, we can deduce that the first component is the most relevant of the system. This component will determine the lifetime of the system with more frequency than the rest of components. Furthermore, for a given lifetime of $X_1$, $m_1(X_1)$ provide the best estimation of the system lifetime.

\textcolor{blue}{\section{Conditions that lead to comparisons in importance}}
\label{cond2compare}
\textcolor{blue}{Sometimes, it is more important to study the relative importance of the components than the corresponding values of the importance measures. It is clear that $R_i^2\leq R_j^2$ holds if, and only if, $E[(m_i(X_i))^2]\leq E[(m_j(X_j))^2]$. However, these expected values are not always simple to compute, and a stochastic comparison of the random variables \(m_i(X_i)\) and \(m_j(X_j)\) can be a more suitable strategy, as suggested in \cite{ShSS} and \cite{ShSS02}. It follows from \eqref{rii}  that if \(m_i(X_i)\) is smaller than \(m_j(X_j)\) with respect to some univariate variability order, then the component \(X_i\) has less influence on $T$ than \(X_j\). Explicitly, let \(\le_{\textup{variability}}\) be a variability order (such as convex order, dispersive order or excess wealth order), then 
\[
m_i(X_i)\le_{\textup{variability}}m_j(X_j) \Longrightarrow  R_i^2 \leq R_j^2. 
\]
Analogously, if \(e_i(X_i)\) is larger than \(e_j(X_j)\)  with respect to some univariate stochastic order of magnitude, then the component \(X_i\) has less influence on $T$ than \(X_j\). Explicitly, let \(\le_{\textup{magnitude}}\) be a univariate order of magnitude (such as  usual stochastic order, increasing convex order or increasing concave order), then 
\[
e_i(X_i)\ge_{\textup{magnitude}}e_j(X_j) \Longrightarrow  R_i^2 \leq R_j^2. 
\]}
\textcolor{blue}{Next, we recall the stochastic orders that will be used in this section. Firstly, we define the usual stochastic order ($\leq_{st}$) and the convex order ($\leq_{cx}$) for the univariate case. Secondly, we define the usual stochastic order for the bivariate case. For properties and theoretical results related with these stochastic orders see \cite{Sh1}. In the sequel, ``increasing'' and ``decreasing'' stand for ``nondecreasing'' and ``nonincreasing'', respectively. The symbol `\(=_{\textup{st}}\)' denotes equality in law.} 

\textcolor{blue}{\begin{Def}
Let $X$ and $Y$ be two non-negative random variables with reliability functions $\bar{F}$ and $\bar{G}$, respectively. 
\begin{enumerate}
\item We say that $X$ is smaller than $Y$ in the usual stochastic order (denoted by $X\leq_{st}Y$) if, and only if, $\bar{F}(x)\leq \bar{G}(x)$ for all $x\in \IR$.
\item $X$ is smaller than $Y$ in the convex order (denoted by $X\leq_{cx}Y$) if, and only if, $E[\phi(X)]\leq E[\phi(Y)]$ for all convex functions $\phi: \IR\rightarrow \IR$ for which the expectations exist.
 \end{enumerate}\end{Def}
\begin{Def}
Given  two bidimensional random vectors $\boldsymbol{X}$ and $\boldsymbol{Y}$, we say that $\boldsymbol{X}$ is smaller than $\boldsymbol{Y}$ in the usual stochastic order (denoted by $\boldsymbol{X}\le_{st}\boldsymbol{Y}$) if $E[\phi(\boldsymbol{X})] \leq E[\phi(\boldsymbol{Y})]$ for all increasing real-valued functions $\phi$ for which these expectations exist.
\end{Def}}

 \textcolor{blue}{To compare the random variables $m_i(X_i)$ and $m_j(X_j)$, in terms of the convex order, we will need to require the strictly monotonic condition to the functions $m_i(x)$ and $m_j(x)$. Firstly,} 
 we start analysing the function $m_i(x)$. Intuitively, one could think that $m_i(x)=E[T|X_i=x]$ is an increasing function, however this is not always true. Let us consider a series system with lifetime $T = \min(X_1,X_2)$, where $X_1$ and $X_2$ are two dependent components. Let us assume that the dependence structure is modelled by the FGM copula of dimension 2:

\begin{equation}
\label{copulaFGM2}
\textcolor{blue}{C(u_1,u_2)} = u_1\,u_2\,[1+\theta\,(1-u_1)\,(1-u_2)], \,\,\,\, u_1,u_2\in[0,1].
\end{equation}

\textcolor{blue}{Note that the FGM copula is radially symmetric, i.e., $\hat{C}(u_1,u_2)=C(u_1,u_2)$ for all $u_1,u_2\in [0,1]$, where $\hat{C}$ is the corresponding survival copula.} Thus, from Lemma \ref{l1} we have that 

\begin{equation*}
m_1(x) = E[T|X_1=x] = \int_0^x \hat{C}_{1,P}(\bar F_1(x),\bar F_2(t) ) dt.
\end{equation*}

Hence, in the case of having exponentially distributed components, $X_i\sim Exp(\lambda_i)$ for $i=1,2$, we obtain that 

\begin{equation*}
m_1(x) = \frac{1-\exp(-\lambda_2\,x)}{\lambda_2}+\frac{\theta}{\lambda_2}(1-2\exp(-\lambda_1\,x))\left [1-\exp(-\lambda_2\,x)-\frac{1}{2}(1-\exp(-2\lambda_2\,x))\right].
\end{equation*}

Figure \ref{fig_Cexmple} shows the plot of $m_1$ (red line) in the case of having $\theta=-1$ in \eqref{copulaFGM2} (negative dependence) and $\lambda_i = i$ for $i=1,2$, the failure rates of $X_1$ and $X_2$, respectively. As we can see in Figure \ref{fig_Cexmple}, $m_1(x)$ is not increasing for all $x$. Note that for $\theta=-1$, the FGM copula does not express a significant discordance. Indeed, the Kendall's tau takes the value $-2/9$, see Example 5.2. in \cite{Ne06}.

\begin{figure}[t!]
	\begin{center}
		\includegraphics*[scale=0.4]{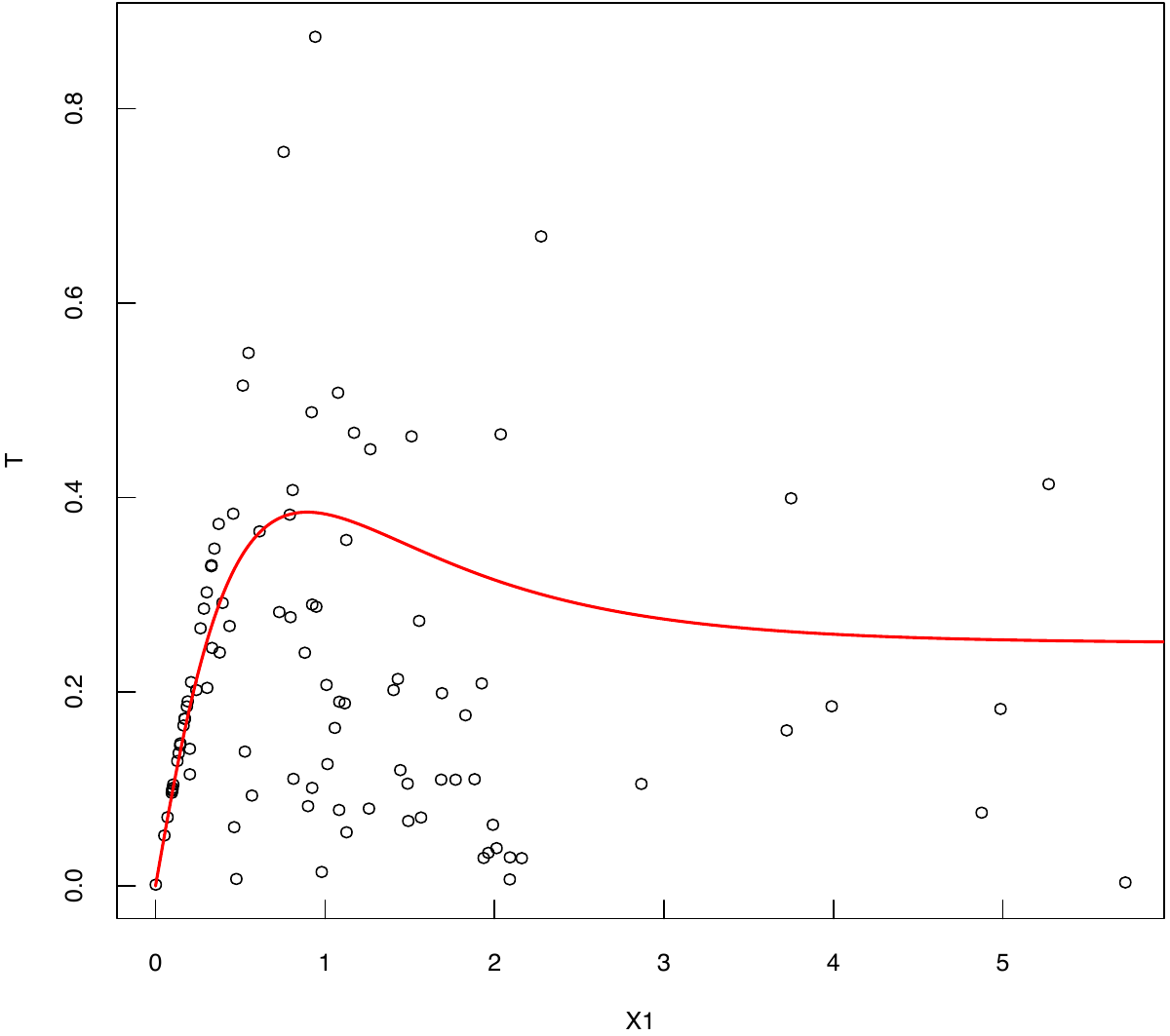}
		\caption{Plot of $m_1$ (red line) jointly with $100$ simulated data.} \label{fig_Cexmple}
	\end{center}
\end{figure}

Next, we provide some conditions to ensure the monotonicity of $m_1$. To do that, we recall first some definitions related with the idea of positively dependent structure in a random vector. The following dependence structure was initially called as positive regression dependence by Lehmann in \cite{Leh}. However, we will refer to this concept as stochastically increasing, following the terminology used by Shaked in \cite{Sh77}.

\begin{Def}
Let $X$ and $Y$ be two random variables, we say that $Y$ is stochastically increasing (SI) in $X$ if
$$\{Y|X=x_1\}\leq_{st} \{Y|X=x_2\}$$
\noindent for all $x_1,x_2\in \IR$ such that $x_1\leq x_2$.
\end{Def}

The notion SI can be generalized to random vectors $(X_1,X_2,\ldots,X_n)$ in different ways, see page 21 in \cite{J2}. \textcolor{blue}{Here, we consider the notions conditionally increasing in sequence (CIS) and conditionally increasing (CI) as natural extensions of SI for the multivariate case.} 

\begin{Def}
The random vector $(X_1,X_2,\ldots,X_n)$ is conditionally increasing in sequence if, for $i = 2,\ldots, n$, 
$$\{X_i|X_1=x_1,\ldots,X_{i-1}=x_{i-1}\}  \leq_{st}   \{X_i|X_1=x_1^{\prime},\ldots,X_{i-1}=x_{i-1}^{\prime}\}$$
\noindent for all $x_j\leq x_j^{\prime}$ and $j=1,\ldots, i-1.$
\end{Def}

\begin{Def}
A random vector $(X_1,\ldots,X_n)$ is conditionally increasing (CI) if, and only if, the random vector $(X_{\pi(1)},\ldots,X_{\pi(n)})$ is CIS for all permutation $\pi\in \Pi_n$.
\end{Def}

The CI concept was introduced by M\"uller and Scarsini in \cite{mull0}. The following result provides a sufficient condition to ensure the monotonicity of $m_i(x).$

\begin{The}
\label{mi_increasing}
Let $T$ be the lifetime of a coherent system with $n$ components and let $X_i$ be the lifetime of the $i$th component for $i=1,\ldots,n$. If the random vector $(X_1,\ldots,X_n)$ is absolutely continuous CI, then $m_i(x)=E[T|X_i=x]$ is increasing for all $i=1,2,\ldots,n$.
\end{The}

\begin{proof}
It is well-known that the lifetime of the system can be written as $T = \psi(X_1,\ldots, X_n)$ where $\psi:\IR^n\rightarrow \IR$ is an increasing real function. If the random vector $(X_1,\ldots,X_n)$ is absolutely continuous CI, then from Corollary 11 in \cite{Sord}, we obtain that $\psi(X_1,\ldots, X_n)$ is SI in $X_i$ for all $i = 1,\ldots,n$, i.e., $\{T|X_i=x\}\leq_{st} \{T|X_i=x^{\prime}\}$ for all $x\leq x^{\prime}$ and $i = 1,\ldots,n$. Therefore, $m_i(x)$ is increasing for all x and $i=1,2,\ldots,n$. 
\end{proof}

A sufficient condition for CI is the property MTP$_2$, investigated by Karlin and Rinott in \cite{Kar}. If a copula function satisfies the properties CIS, CI or MTP$_2$, then a random vector with dependence structure given by that copula inherits the same property, see Proposition 3.5 in \cite{mull0}. There exist many examples of copulas which satisfy either the CI or MTP$_2$ condition. The family of Archimedean copulas are CI. Specifically, the independence copula is MTP$_2$. Cerqueti and Lupi propose in \cite{Cerqueti} a new family of non-exchangeable Archimedean copulas which leads to an MTP$_2$ property. Further examples can be found in \cite{mull0}, \cite{Kar} and \cite{ShSpzz}. 

\begin{Rem}
Under the assumptions of Theorem \ref{mi_increasing}, we cannot ensure that $m_i(x)$ is strictly increasing. Indeed, if we consider $T=\min(X_1,X_2)$ the lifetime of a series system with two independent components and we assume that $X_1$ and $X_2$ are interval valued random variables given by two uniform distributions $U(0,n)$ and $U(0,m)$, respectively, with $n<m$. Then
$$\bar{F}_{T|X_2=x}(t)=\left\{\begin{array}{ccc}
\bar{F}_1(t) & & 0\leq t < n,\\
0            & & t\geq n\\ 
\end{array}\right.$$
\noindent for all $x\in [n,m]$. Note that $\bar{F}_{T|X_2=x}(t)$ does not depend on $x$. Thus, $m_2(x) = E[T|X_2=x]$ is constant for all $x\in [n,m]$.\\

The following result ensures that $m_i(x)$ is strictly increasing under some assumptions. We provide first a remark.

\begin{Rem}
\label{lemmath}
If a random vector $(X_1,\ldots,X_n)$ is absolutely continuous with support $[0,\infty)^n$ then, any conditional random vector of $n-1$ components $(X_1,\ldots,X_{i-1},X_{i+1},\ldots,X_n|X_i=t)$ with $t\geq 0$ is absolutely continuous with support in $[0,\infty)^{n-1}$. 
\end{Rem}

\begin{The}
\label{mainTH}
Let $T=\psi\left(X_{1}, \ldots, X_{n}\right)$ be the lifetime of a coherent system based on $n$ absolutely continuous components with joint support $[0,\infty)^n$. \textcolor{blue}{If the copula $C$ of the random vector $(X_1,\ldots,X_n)$ is CI and the corresponding survival copula $\hat{C}$ have continuous partial derivatives,} then $m_i(x)$ is strictly increasing for all $x\geq 0$. 
\end{The}
\begin{proof}
Let us assume that $m_i(a) = m_i(b)$ for two values $a,\,b\geq 0$ such that $a<b$. The general expression of the survival function of $\{T|X_i=x\}$ is given by
$$
\bar{F}_{T|X_i=x}(t)=\left\{\begin{array}{ccc}
\phi_{i,x}(t) & & 0\leq t < x,\\
\phi_{i,x}^{\star}(t) & & t\geq x,\\ 
\end{array}\right.
$$ 

\noindent where $\phi_{i,x}(t)$ and $\phi_{i,x}^{\star}(t)$ represent the survival functions of the system when the $i$th component is operative and no operative, respectively. Assume that $P_1,P_2,\ldots, P_r$ are the minimal path sets of the coherent system. Then
\begin{align*}
\phi_{i,x}(t)&= \displaystyle \sum_{j_1=1}^{r}\bar{F}_{X_{P_{j_1}}|X_i=x}(t)-\displaystyle \sum_{j_1<j_2}\bar{F}_{X_{P_{j_1}\cup P_{j_2}}|X_i=x}(t)+\ldots +(-1)^{r+1}\bar{F}_{X_{P_1\cup \dots \cup P_r}|X_i=x}(t)\nonumber\\             &= \displaystyle \sum_{j_1=1}^{r}\widehat{C}_i(\bar{F}_1(z_1^{j_1}),\ldots,\bar{F}_i(x),\ldots,\bar{F}_n(z_n^{j_1}))-\displaystyle \sum_{j_1<j_2}\widehat{C}_i(\bar{F}_1(z_1^{j_1,j_2}),\ldots,\bar{F}_i(x),\ldots,\bar{F}_n(z_n^{j_1,j_2}))\nonumber \\
              &\,\,\,\,\,\,+\ldots+(-1)^{r+1} \widehat{C}_i(\bar{F}_1(z_1^{1,2,\ldots,r}),\ldots,\bar{F}_i(x),\ldots,\bar{F}_n(z_n^{1,2,\ldots,r})),
\end{align*}
with $\widehat{C}_i=\partial_i \widehat{C}$ and
$$
z_k^{j_1,j_2,\ldots,j_q}=\left\{\begin{array}{ccc}
t & & k\in P_{j_1}\cup \cdots \cup P_{j_q},\\
0 & & k\notin P_{j_1}\cup \cdots \cup P_{j_q}\\
\end{array}\right.
$$ 
\noindent for all $k\in \{1,\ldots,n\}\setminus \{i\}$ and $q\in \{1,\ldots, r\}$. Note that $\phi_{i,x}(t)$ is a continuous function. Similarly, $\phi_{i,x}^{\star}(t)$ can be expressed as the function $\phi_{i,x}(t)$ but taking only the minimal path sets which do not contain the $i$th component. Thus, we conclude that $\bar{F}_{T|X_i=x}(t)$ is continuous for all $t\in[0,x)\cup(x,\infty)$. However, $\bar{F}_{T|X_i=x}(t)$ is not continuous for $t=x$. To prove that, we note that $\Pr(T=x|X_i=x)>0$ for all $x\geq 0$ holds, from Remark \ref{lemmath} and using that the $i$th component is relevant. On the other hand,  $\displaystyle \lim_{t\rightarrow x^{-}}\bar{F}_{T|X_i=x}(t)-\bar{F}_{T|X_i=x}(x)=\Pr(T=x|X_i=x)$. Therefore, $\bar{F}_{T|X_i=x}(t)$ has a discontinuity of the first kind at $t=x$. Finally, if $(X_1,\ldots,X_n)$ is CI, then $\bar{F}_{T|X_i=a}(t)\leq \bar{F}_{T|X_i=b}(t)$ for all $a\leq b$ and $t\geq 0$. In particular, if $m_i(a)=m_i(b)$ for some $a<b$, that is, $E(T|X_i=a)=E(T|X_i=b)$,  then  $\bar{F}_{T|X_i=a}(t)= \bar{F}_{T|X_i=b}(t)$ for all $t\geq 0$.  This fact means that, $\bar{F}_{T|X_i=a}$ posses two different discontinuity points at $t=a$ and $t=b$, which is not possible. 
\end{proof}
\end{Rem}

Now, we are able to provide a sufficient condition to compare $R_i^2$ and $R_j^2$. Firstly, we recall the definition of $S^{-}(f)$ the sign changes of a function $f$ on a subset $I\subseteq \IR$:
$$S^{-}(f)=\sup\{S^{-}(f(x_1),\ldots,f(x_k))\},$$
\noindent where $S^{-}(f(x_1),\ldots,f(x_k))$ denotes the sign changes of the indicated sequence, and the supremum is extended over all sets $x_1<\ldots<x_k$ such that $x_i\in I$ and $k\in \IN$.
\begin{The}
Let $T=\psi(X_1,\ldots,X_n)$ be the lifetime of a coherent system with $n$ components having the joint support $[0,\infty)^n$. Let us consider the quantile functions $F^{-1}_{m_i}$ and $F^{-1}_{m_j}$ of the random variables $m_i(X_i)$ and $m_j(X_j)$, respectively. If the random vector $(X_1,\ldots,X_n)$ is absolutely continuous and CI, with associated survival copula $\hat{C}$ having continuous partial derivatives and $S^{-}(F^{-1}_{m_i}-F^{-1}_{m_j})=1$ with sign sequence $+,\,-$ for $i,j\in \{1,2,\ldots,n\}$ with $i\neq j$, then 
\begin{equation*}
 m_i(X_i)\le_{cx} m_j(X_j).
\end{equation*}
In particular, $R_i^2\leq R_j^2$ holds.
\end{The}
\begin{proof}
Under the assumptions of the theorem and from Theorem \ref{mainTH}, we get that $m_i(x)$ and $m_j(x)$ are continuous and strictly increasing for all $x\geq 0$. Then, we can express the quantile functions of $m_i(X_i)$ and $m_j(X_j)$ as $F^{-1}_{m_i}(p) = m_i(F_i^{-1}(p))$ and $F^{-1}_{m_j}(p) = m_j(F_j^{-1}(p))$, respectively,  where $F^{-1}_{i}$ and $F^{-1}_{j}$ are the quantile functions of $X_i$ and $X_j$. If $S^{-}(F^{-1}_{m_i}-F^{-1}_{m_j})=1$ with sign sequence $+,\,-$, then $S^{-}(F_{m_j}-F_{m_i})=1$ with sign sequence $+,\,-$. Taking into account that $E[m_i(X_i)]=E[m_j(X_j)]=E[T]$, then from Theorem 3.A.44 in \cite{Sh1} we conclude that $m_i(X_i)\leq_{cx}m_j(X_j)$.  

\end{proof}

Before provide our next result, we first recall the definition of the concordance order (see Definition 2.8.1 in \cite{Ne99}).

\begin{Def}
	Given two copulas $C$ and $C^{\prime },$ we say that $C$ is smaller than $%
	C^{\prime }$ in the concordance order (denoted by $C\prec C^{\prime })$ if $%
	C\left( u,v\right) \leq C^{\prime }\left( u,v\right) $ for all   $u,v \in \left( 0,1\right)$.
\end{Def}
	
\begin{The}
\label{th_concord}
 Let $T=\psi\left(X_{1}, \ldots, X_{n}\right)$ be the lifetime of a coherent system based on $n$ absolutely continuous components with common distribution function $F$ and joint support $[0,\infty)^n$. Assume that the vector $(X_{1}, \ldots, X_{n} )$ is CI and denote by $C^{(k)}$ the copula of the vector $(T,X_k), \ k=1,...,n.$ Then, $C^{(i)}\prec C^{(j)}$ implies \begin{equation}\label{convex}
 m_i(X_i)\le_{cx} m_j(X_j) 
 \end{equation}
 for $i,j\in \{1,2,\ldots,n\}$ with $i\neq j.$ In particular, $R_i^2\leq R_j^2$ holds.
\end{The}	
\begin{proof}
	Let $F_{m_k}$ be the distribution function of the random variable  $m_k(X_k)=E\left[T \mid X_{k}\right],$ $k=1,...,n.$ By using Theorem \ref{mainTH}, it follows that the function $m_k$ is strictly increasing,  $k=1,...,n.$ Then, we have
		$$F_{m_k}(t) = P\left[ m_{k}(X_k) \leq t \right] = P\left[ X_{k} \leq m_k^{-1} (t) \right] = 
	F(m^{-1}_k (t)), \ \forall t.$$
The quantile function of $m_{k}(X_k)$  is given by	
	$$F_{m_k}^{-1}(p) =m_k(F^{-1}(p) )= E\left[ T \mid X_{k} =F^{-1}(p)
\right], \ p\in(0,1).$$
Note that $E[m_k(X_k)]=E[T]$ for $k=1,...,n.$  Given $i,j$ such that $i,j\in \{1,2,\ldots,n\}$ with $i\neq j,$ it follows from  Theorem 3.5 in Shaked and Shantikhumar (2007) that \eqref{convex} is equivalent to 	
	$$\int_{p}^{1} F^{-1}_{m_i}(x) \ dx \leq 
	\int_{p}^{1} F^{-1}_{m_j}(x)  \ dx \,\,\mbox{for all}\,\,  p \in (0,1),$$
or equivalently,
$$\int_{p}^{1} E\left[ T \mid X_{i} =F^{-1}(x)\right]   \ dx \leq 
	\int_{p}^{1}  E\left[ T \mid X_{j} =F^{-1}(x)\right]  \ dx \,\,\mbox{for all}\,\,  p \in (0,1),$$
that is,
\begin{equation} \label{lexp}
E\left[ T \mid X_{i} \ge F^{-1}(p)\right] \le E\left[ T \mid X_{j} \ge F^{-1}(p)\right] \,\,\mbox{for all}\,\,  p \in (0,1).
\end{equation}
Denote by $\bar{F}_T(t)$ the survival function of the system $T.$ Given $0<p<1$ and $k=1,\ldots,n,$ we can write
\begin{equation}
\label{hkf}
E\left[ T \mid X_{k} \ge F^{-1}(p)\right]=\int_{0}^{1} P[T >t \mid X_k \ge F^{-1}(p) ] dt = \int_{0}^{1}h_k(\bar{F}_T(t))dt,  
\end{equation}
where
\begin{equation*}\label{distk}
h_k(u)=\frac{u-p+C^{(k)}(1-u, p)}{1-p}, \ 0\le u \le 1,
\end{equation*}
is a distortion function. It is clear that $C^{(i)}\prec C^{(j)}$ implies $h_i(u) \le h_j(u)$ for all $u \in (0,1)$. From this fact and \eqref{hkf}, we see that $C^{(i)}\prec C^{(j)}$  implies \eqref{lexp}, which is the same as \eqref{convex}.
\end{proof}

\begin{Exa}
Let us consider a series system with two independent and exponentially distributed components,  $X_1\sim Exp(\lambda_1)$ and $X_2\sim Exp(\lambda_2)$. Firstly, we calculate the survival copula associated to the vector $(T, X_1)$ where $T = \min(X_1,X_2)$.

\begin{align*}
\widehat{C}^{(1)}(u,v)&= \Pr[T>F^{-1}_T(1-u), X_1>F_1^{-1}(1-v)] = \\
            &= \Pr[\min(X_1,X_2)>F^{-1}_T(1-u), X_1>F_1^{-1}(1-v)]\\
            &= \Pr[X_1>\max(F_1^{-1}(1-v),F^{-1}_T(1-u)), X_2>F^{-1}_T(1-u)]\\
            &=\left\{\begin{array}{lcc}
                v\,\bar{F}_2(F^{-1}_T(1-u)) &\mbox{if}& F_1^{-1}(1-v)\geq F^{-1}_T(1-u),  \\
                \bar{F}_1(F^{-1}_T(1-u))\,\bar{F}_2(F^{-1}_T(1-u)) &\mbox{if}& F_1^{-1}(1-v)< F^{-1}_T(1-u).  \\
              \end{array}\right.  
\end{align*}

Taking into account that $\bar{F}_T(t)=e^{-(\lambda_1+\lambda_2)\,t}$, $F_1^{-1}(1-v) = \dfrac{-1}{\lambda_1} \log(v)$, $F_T^{-1}(1-u) = \dfrac{-1}{\lambda_1+\lambda_2} \log(u)$ and $\bar{F}_i(t)=e^{-\lambda_i\,t}\,\, \mbox{for}\,\, i=1,2.$ We obtain that 
$$\widehat{C}^{(1)}(u,v) =\left\{\begin{array}{ccc}
                v\,u^{\frac{\lambda_2}{\lambda_1+\lambda_2}} &\mbox{if}& v\leq u^{\frac{\lambda_1}{\lambda_1+\lambda_2}},  \\
                u &\mbox{if}& v> u^{\frac{\lambda_1}{\lambda_1+\lambda_2}}.  \\
              \end{array}\right.$$
              
Similarly, we calculate the survival copula for the vector $(T, X_2)$,
$$\widehat{C}^{(2)}(u,v) =\left\{\begin{array}{ccc}
                v\,u^{\frac{\lambda_1}{\lambda_1+\lambda_2}} &\mbox{if}& v\leq u^{\frac{\lambda_2}{\lambda_1+\lambda_2}},  \\
                u &\mbox{if}& v> u^{\frac{\lambda_2}{\lambda_1+\lambda_2}}.  \\
              \end{array}\right.$$

We study now the sign of $\hat{C}^{(2)}(u,v)-\hat{C}^{(1)}(u,v)$ assuming that  $\lambda_1< \lambda_2$.\\

\textbf{Case I:} $u^{\frac{\lambda_2}{\lambda_1+\lambda_2}}<u^{\frac{\lambda_1}{\lambda_1+\lambda_2}}<v$.\\

$\hat{C}^{(2)}(u,v)-\hat{C}^{(1)}(u,v) = u-u = 0$.\\

\textbf{Case II:} $u^{\frac{\lambda_2}{\lambda_1+\lambda_2}}<v\leq u^{\frac{\lambda_1}{\lambda_1+\lambda_2}}$.\\

$\hat{C}^{(2)}(u,v)-\hat{C}^{(1)}(u,v) = u-vu^{\frac{\lambda_2}{\lambda_1+\lambda_2}} = u^{\frac{\lambda_2}{\lambda_1+\lambda_2}}(u^{\frac{\lambda_1}{\lambda_1+\lambda_2}}-v)\geq 0$.\\

\textbf{Case III:} $v \leq u^{\frac{\lambda_2}{\lambda_1+\lambda_2}}<u^{\frac{\lambda_1}{\lambda_1+\lambda_2}}$.\\

$\hat{C}^{(2)}(u,v)-\hat{C}^{(1)}(u,v) = vu^{\frac{\lambda_1}{\lambda_1+\lambda_2}}-vu^{\frac{\lambda_2}{\lambda_1+\lambda_2}} = v(u^{\frac{\lambda_1}{\lambda_1+\lambda_2}}-u^{\frac{\lambda_2}{\lambda_1+\lambda_2}}) \geq  0$.\\

We conclude that $\hat{C}^{(2)}(u,v)-\hat{C}^{(1)}(u,v)\geq 0$ for all $u,v\in (0,1)$, i.e., $\hat{C}^{(1)}\prec \hat{C}^{(2)}$ and, therefore, from Theorem \ref{th_concord} we get that $R_1^2\leq R_2^2$ for all $\lambda_1<\lambda_2$. Note that this conclusion was reached in Remark \ref{obsR2seriesystem}. However, here we have also proved that $m_1(X_1)\le_{cx} m_2(X_2), $ which is a stronger ordering between $m_1(X_1)$ and $m_2(X_2)$.
  
\end{Exa}

Now we focus on the bivariate random vector $(T,X_k),$ for $k=1,...,n$. Let $\boldsymbol{I}_k$ the random vector defined by
\begin{equation}\label{DSM}
\boldsymbol{I}_k=(i, j) \quad \text { whenever } \quad T=X_{i: n} \quad \text { and } \quad X_k=X_{j: n}.
\end{equation}
The bivariate probability mass function of $\boldsymbol{I}_k$ is denoted by $p^k_{i, j}=P\left[\boldsymbol{I}_k=(i, j)\right],$ for $i, j=$ $1, \ldots, n$ (of  course, $\sum_{i=1}^{n} \sum_{j=1}^{n} p^k_{i, j}=1$). The matrix $\boldsymbol{P}_k=\left(p^k_{i, j}\right)$ is called the bivariate signature matrix associated with $\left(T, X_k\right)$, see \cite{Navarroetal2013}. It can be shown that
$$
p^k_{i, j}=\frac{\left|A^k_{i, j}\right|}{n !},
$$
where $\left|A^k_{i, j}\right|$ is the cardinality of the set
$$
A^k_{i, j}=\left\{\sigma \in \mathscr{P}_{n}: T=X_{i: n} \text { and } X_k=X_{j: n} \text { whenever } X_{\sigma(1)}<\cdots<X_{\sigma(n)}\right\}
$$
and $\mathscr{P}_{n}$ is the set of permutations of the set $\{1, \ldots, n\}.$ 

The following result shows that for a coherent system $T$ with independent and identically distributed components $X_1,...,X_n,$ the stochastic ordering (as discrete distributions) of the bivariate signatures $\boldsymbol{I}_i$ and $\boldsymbol{I}_j,$ associated to $(T,X_i)$ and $(T,X_j)$ respectively, is a sufficient condition for \eqref{convex}.
\begin{Cor}
\label{iid}
 Let $T=\psi\left(X_{1}, \ldots, X_{n}\right)$ be the lifetime of a coherent system based on $n$ absolutely continuous i.i.d. component lifetimes with a common distribution function $F.$  Let $\boldsymbol{I}_k$ be the random vector defined by \eqref{DSM} for $\left(T, X_k\right)$, $1\le k\le n. $ If $\boldsymbol{I}_i \leq_{st} \boldsymbol{I}_j$, then $m_i(X_i)\leq_{cx} m_j(X_j)$.
\end{Cor}
\begin{proof}
Under the assumptions, it follows from Theorem 3.1 in \cite{Navarroetal2013} that $(T,X_i)\le_{st}(T,X_j)$. This implies (see \cite{Sh1}, p. 308) that 
$$
\Pr\left(T\leq t, X_i \leq s\right) \geq \Pr\left(T \leq t, X_j \leq s\right)\text{ for all } s>0,t>0,
$$
which is the same as $C^{(i)}\prec C^{(j)}$ by using that $(T,X_i)$ and $(T,X_j)$ have the same marginal distribution functions. Now the result follows by applying Theorem \ref{th_concord}.

\end{proof}

\begin{Rem}
	Corollary \ref{iid} can be extended to the case of a system with exchangeable components using Theorem 3.2 in \cite{Navarroetal2013} instead of Theorem 3.1 in \cite{Navarroetal2013}.
\end{Rem}

Navarro et al. (\cite{Navarroetal2013}, p. 1021-22) analyzed the conditions on two bivariate signatures $\boldsymbol{I}$ and $\boldsymbol{I^*}$ to have $\boldsymbol{I}\le_{st}\boldsymbol{I^*}$. In particular, if $\boldsymbol{P}_k$ and  $\boldsymbol{P}_{k'}$ are the signature matrices associated to $(T,X_k)$ and $(T,X_{k'}),$ respectively, the condition  $\boldsymbol{I}_k \leq_{st} \boldsymbol{I}_{k'}$ is equivalent to obtain $\boldsymbol{P}_{k'}$ from $\boldsymbol{P}_k$ through a finite sequence of transformations in which a positive mass $c>0$ is moved from the term $p_{i, j}^{k}$ to the term $p_{r, s}^{k'}$ with $r \geq i$ and $s \geq j$ (i.e. the new terms are $p_{i, j}^{k}-c$ and $p_{r, s}^{k'}+c$, respectively).

\section{\textcolor{blue}{Computation of the importance measure by simulation}}
\label{estimationIndex}

\textcolor{blue}{This section provides a procedure to approximate the importance measure by Monte Carlo simulation methods. Firstly, we would like to note that it is not easy to take samples of components' lifetimes when dependence exists among them. However, in the case of having lifetimes' data, we could estimate the distribution of the components, the corresponding copula and the regression curve $m_i(x)$ by the procedure described in \cite{Noh}. Observe that in our case the regression curve satisfies, under some assumptions, the monotonicity property required in \cite{Dette}. Therefore, we will assume that we know the quality of the components (distributions) and the dependence structure (copula) among them. Then, we simulate the lifetimes of the system components (as many as desired) to obtain an approximation of the variance of $T$ and $m_i(X_i)$, given that the expressions of $m_i(x)$ are known for all $i=1,\ldots,n$. If the functions $m_i(x)$ are not available, we can always approximate $m_i(x)$ by numerical integration or by using nearest-neighbor methods, see, for example, page 19 in \cite{Hastieetal2009}.} 
\begin{table}
{\rowcolors{2}{gray!20!white!50}{gray!70!white!30}
\begin{center}
\begin{tabular}{ |p{2.1cm}|p{2cm}|p{2cm}|p{3cm}| }
\hline
Sample size&  \hspace{0.75cm}$\hat{R}_1^2$  & \hspace{0.75cm}$R_1^2$ & Absolute Error\\
\hline 
\hspace{0.5cm}100  & 0.8459399   & 0.872727273 & 0.02678737\\
\hspace{0.5cm}500  & 0.8662224   & 0.872727273 & 0.006504873\\
\hspace{0.5cm}1000 & 0.869915 & 0.872727273 & 0.002812273\\
\hspace{0.5cm}1500   & 0.875313   & 0.872727273 & 0.002585685\\
\hspace{0.5cm}5000   &  0.8722335  & 0.872727273 & 0.0004937255\\
\hline
\end{tabular}
\end{center}}
\caption{\textcolor{blue}{Approximations} of $R_1^2$ in Example \ref{Example34}, depending on the sample size.}
\label{Tab01}   
\end{table}

\textcolor{blue}{To compare the approximated values of $R_i^2$ (denoted by $\hat{R}_i^2$) with the real ones, we consider the system studied in Example \ref{Example34}. If we assume that the components are independent and have standard exponential distributions, then we can generate lifetimes' data for each component with different sample sizes. Fixed a sample size, we can obtain an approximation for $R_1^2$ and $R_2^2$, just by using Monte Carlo method and the expressions in \eqref{m1Ex36} and \eqref{m2Ex36}, respectively. Tables \ref{Tab01} and \ref{Tab02} represent some examples of these approximations for $R_1^2$ and $R_2^2$, respectively. In general, the absolute error decreases as the sample size increases, as expected. We now analyse the distribution of the errors depending on the sample size. We define $E_i = R_i^2-\hat{R}_i^2$ as the approximation errors of the importance measure for the $i$th component with $i=1,2$. Fixed a sample size, we calculate 1000 approximations of the values $R_1^2$ and $R_2^2$, and we compute the corresponding errors. Figure \ref{fig02} represents the errors' distributions depending on the sample size. As we can see, there exists a lack of symmetry in the errors' distributions for the cases of small sample sizes and the errors' dispersion decreases as the sample size increases.} 
\begin{table}
{\rowcolors{2}{gray!20!white!50}{gray!70!white!30}
\begin{center}
\begin{tabular}{ |p{2.1cm}|p{2cm}|p{2cm}|p{3cm}| }
\hline
Sample size& \hspace{0.75cm}$\hat{R}_2^2$ & \hspace{0.75cm}$R_2^2$ & Absolute Error\\
\hline
\hspace{0.7cm}100  & 0.03566282 & 0.02424242 & 0.0114204\\
\hspace{0.7cm}500  & 0.02462768 & 0.02424242 & 0.00038526\\
\hspace{0.7cm}1000 & 0.02334157 & 0.02424242 & 0.00090085\\
\hspace{0.7cm}1500   & 0.02451261 & 0.02424242 & 0.0002701891\\
\hspace{0.7cm}5000   & 0.02419679 & 0.02424242 & 0.0000456361\\
\hline
\end{tabular}
\end{center}}
\caption{\textcolor{blue}{Approximations} of $R_2^2$ in Example \ref{Example34}, depending on the sample size.}
\label{Tab02}   
\end{table}

\begin{figure}[h!]
	\begin{center}
		\includegraphics*[scale=0.3]{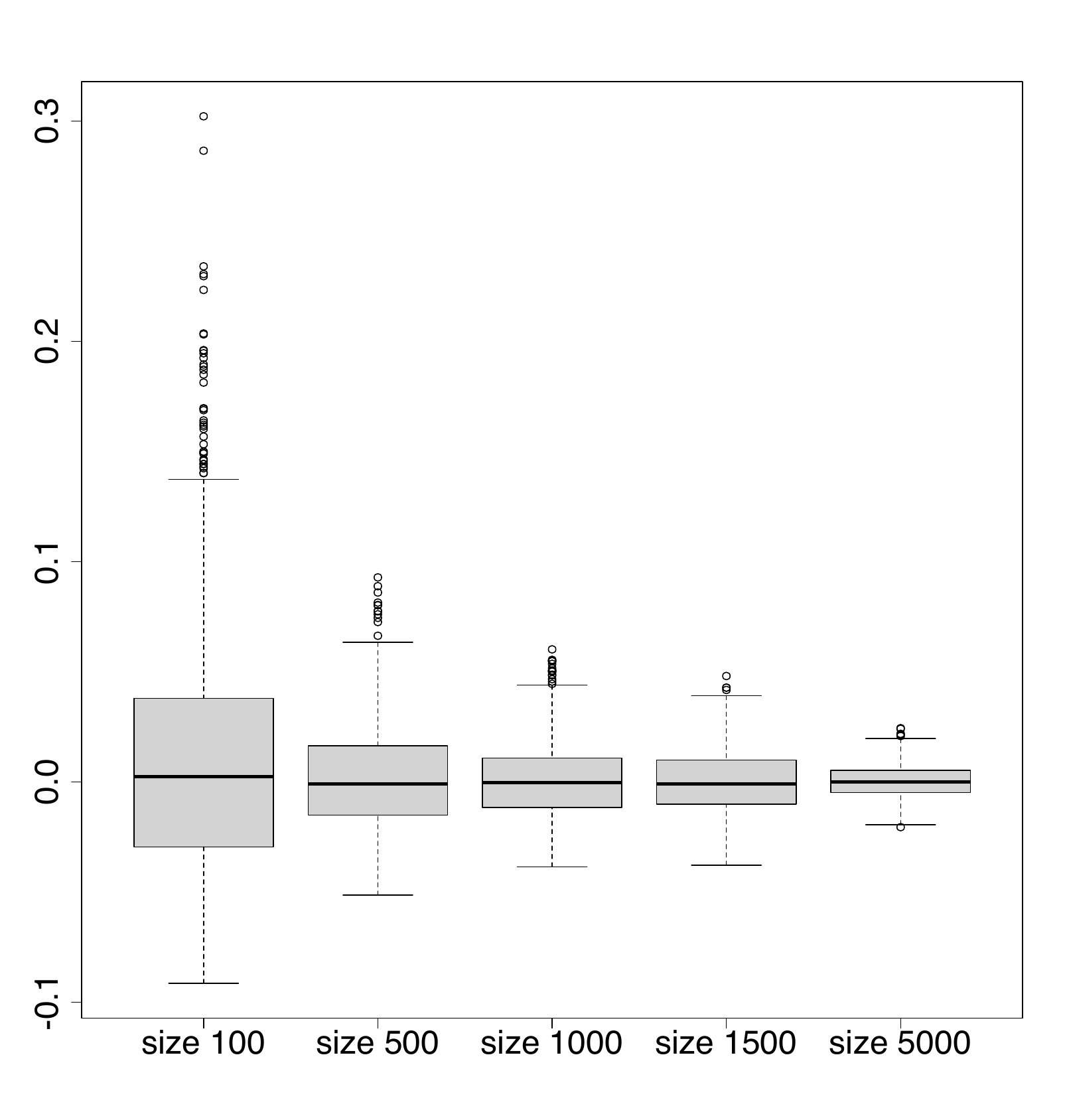}
		\includegraphics*[scale=0.3]{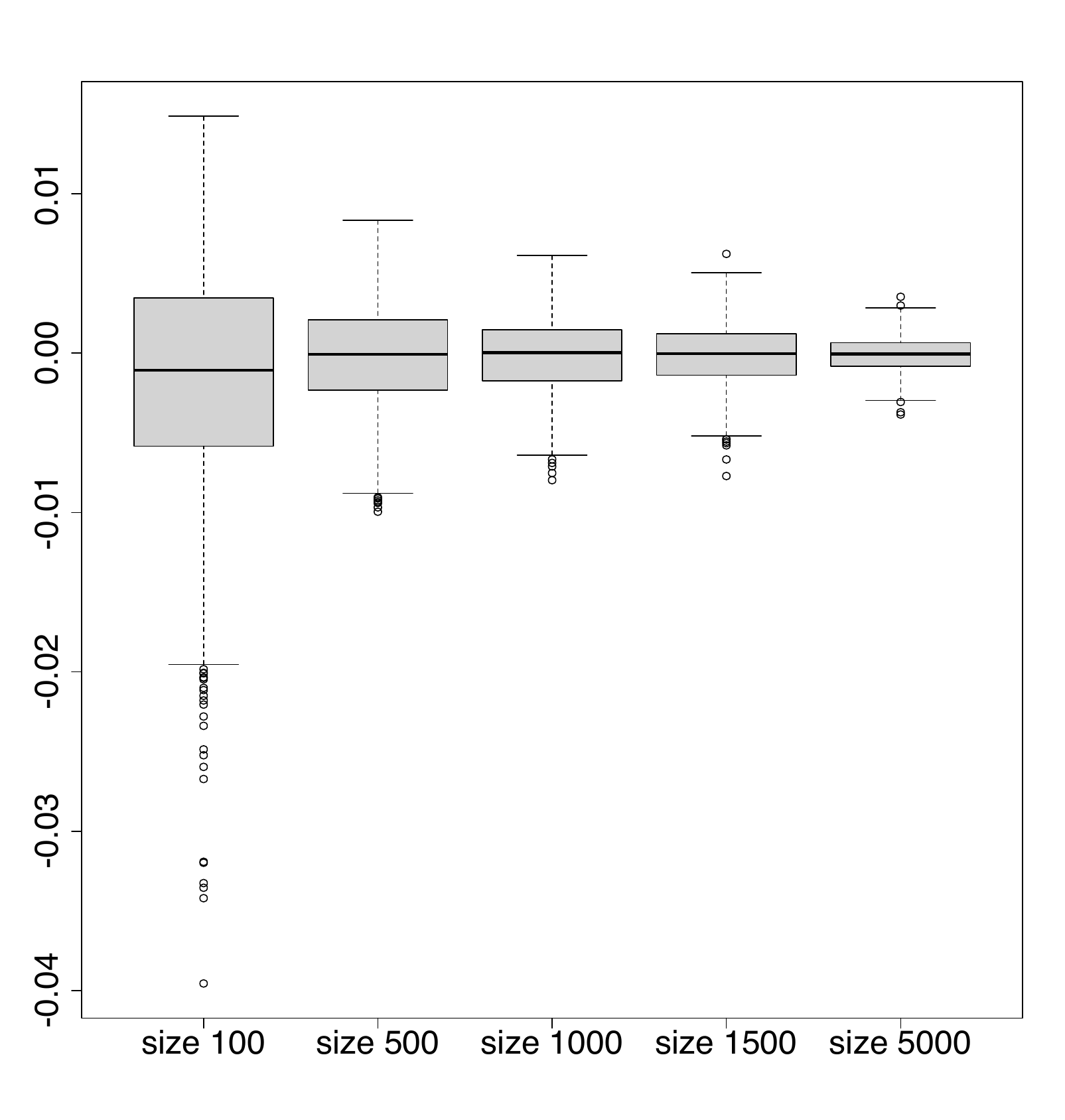}
		\caption{\textcolor{blue}{Plots of the errors' distributions associated to the \textcolor{blue}{approximations} of $R_1^2$ and $R_2^2$ (left and right, respectively) depending on the sample sizes.}} 
		\label{fig02}
	\end{center}
\end{figure}
\end{Exa}

\textcolor{blue}{As an applied example to illustrate the utility of the previous results, we consider a coherent system in the context of naval engineering. Specifically, we study a simplified model of a ship control system with 4 dependent and heterogeneous components. The structure of dependence will be modelled by the Farlie-Gumbel-Morgenstern (FGM) copula.} 
\begin{Exa}
\label{Example51}
\textcolor{blue}{The control system of a ship is carried out through the control panels of the servo (an automatic device that uses error-sensing negative feedback to correct the action of a mechanism). The lifetime associated with the considered system is given by $T = \max(X_1, \min(X_2, X_3), \min(X_2, X_4))$, where the random variable $X_1$ represents the lifetime of the manual control valves located in the engine room, $X_2$ is the lifetime of the electric motor with local control and, finally, $X_3$ and $X_4$ are the lifetimes of the bridge control panel (recall that the bridge of a ship is the room or platform from which the ship can be commanded) and the machine control panel, respectively. In this situation, it is usual to assume that the components are slightly dependent because they share the same marine environment. We will represent the dependence relationship among the components by the following FGM copula:}
\textcolor{blue}{\begin{equation}
\label{CopFGMdim4}
C(u_1,u_2,u_3,u_4)=u_1u_2u_3u_4[1+\theta\,(1-u_1)(1-u_2)(1-u_3)(1-u_4)],
\end{equation}}
\noindent \textcolor{blue}{where the parameter $\theta \in [-1,1]$. Observe that the case of independent components is obtained just taking $\theta=0$. Let us consider that the components' lifetimes are modelled by Weibull distributions. Hence, we study the system under two scenarios:}

\begin{figure}[t!]
	\begin{center}
		\includegraphics*[scale=0.3]{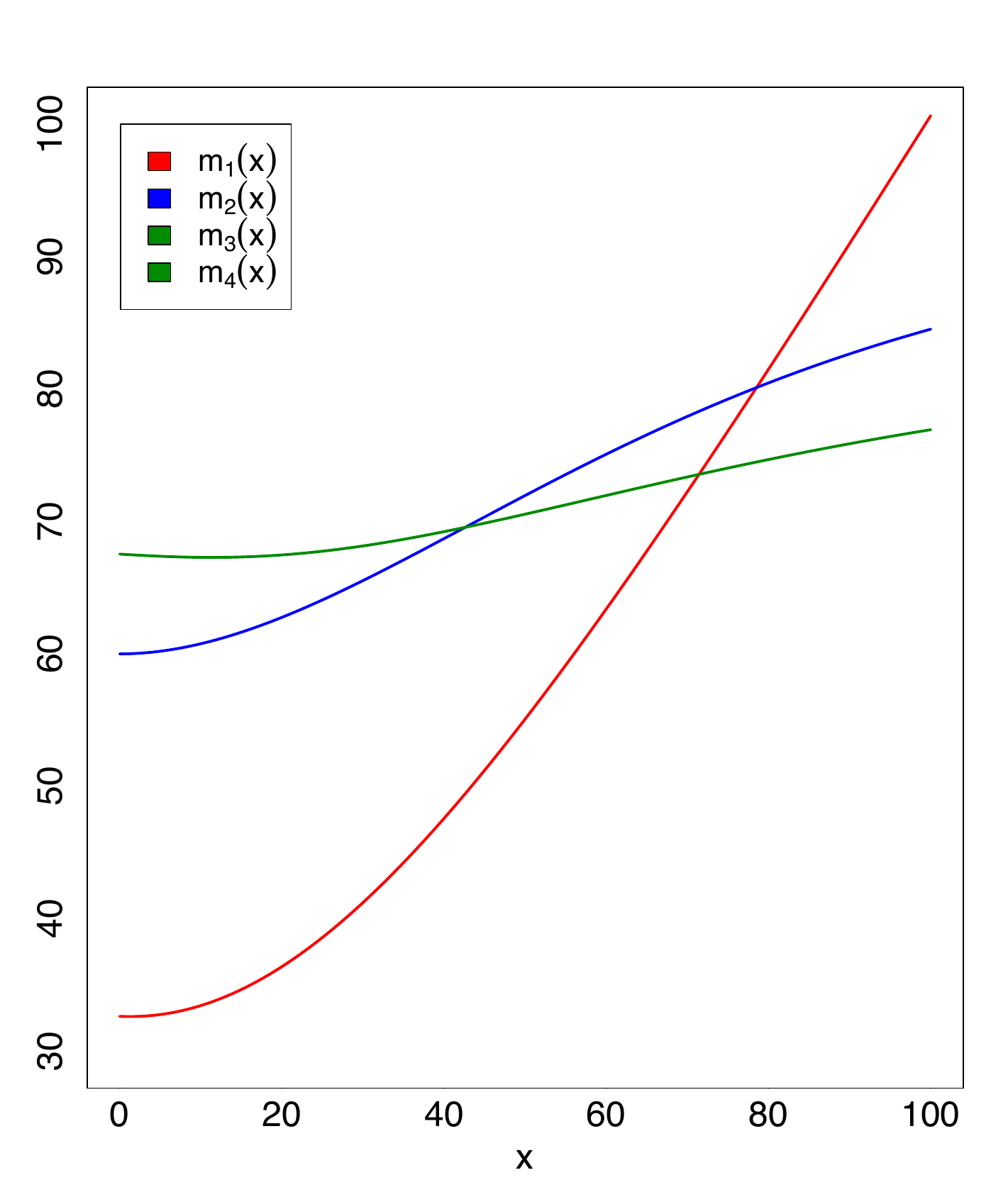}
        \includegraphics*[scale=0.3]{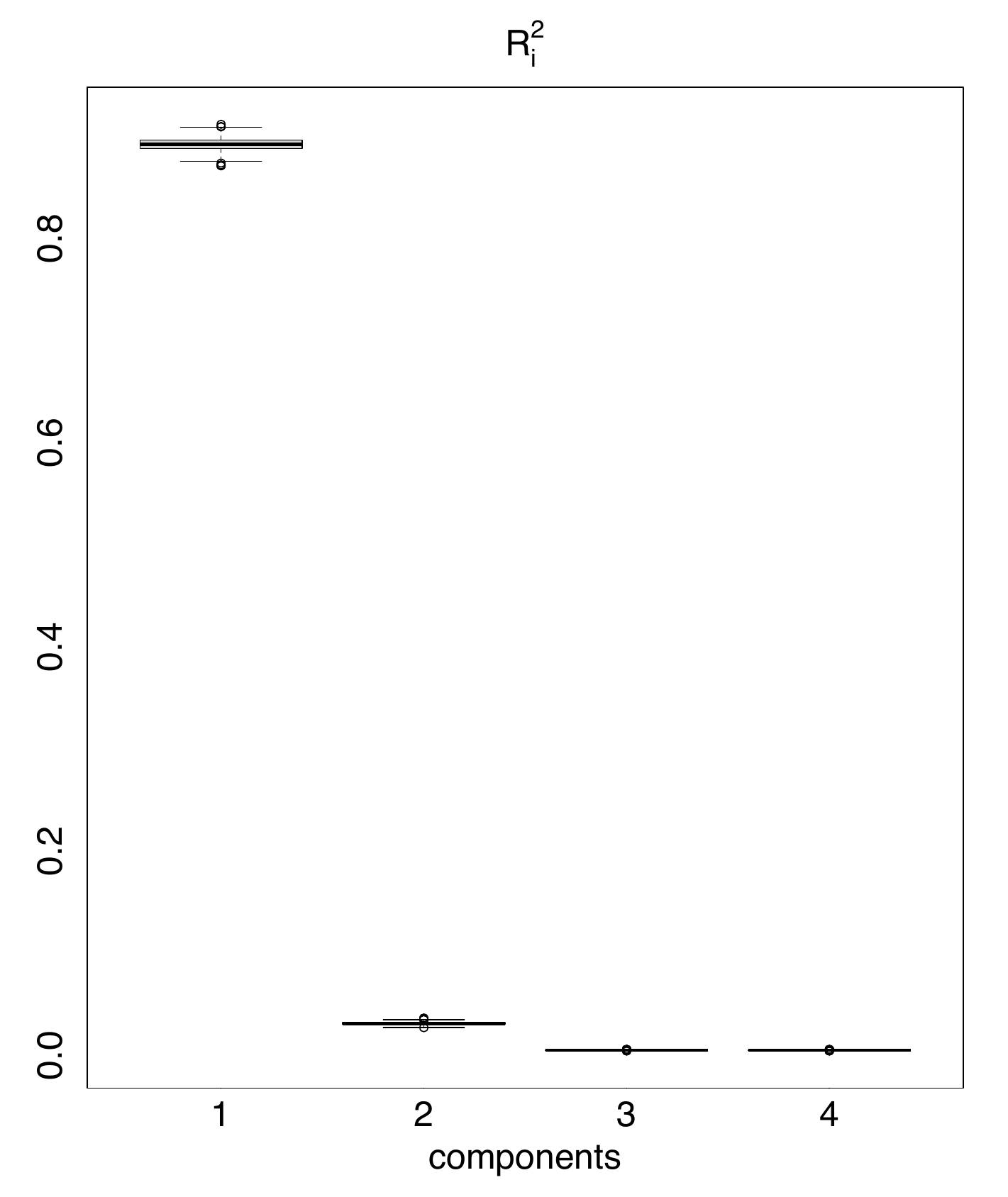}
		\caption{\textcolor{blue}{Plots of $m_1(x)$, $m_2(x)$, $m_3(x)$ and $m_4(x)$ in the case $\theta = 1$ (left). Distributions of the values $\hat{R}_1^2$, $\hat{R}_2^2$, $\hat{R}_3^2$ and $\hat{R}_4^2$ from 1000 approximations in the case $\theta = 1$ (right).}} \label{fig_Cexmple01}
	\end{center}
\end{figure}

\textcolor{blue}{\textbf{Case I:} take $\theta > 0$ in \eqref{CopFGMdim4} (positive dependence) and assume components with constant failure rate functions. Specifically, the components' reliability functions are $\bar{F}_i(t)=\exp(-\lambda_i\,t)$ with $\lambda_1=1/60$, $\lambda_2=1/50$, $\lambda_3=1/45$ and $\lambda_4=1/45$. From the radially symmetric property of the copula \eqref{CopFGMdim4} and Corollary \ref{Cor1}, we obtain that the corresponding regression curves are:}
\begin{align*}
\label{m1Ship}
\hspace{-0.1cm}\textcolor{blue}{m_1(x) =}& \textcolor{blue}{\frac{900}{19}\exp\Bigl(-\frac{19x}{450}\Bigr)+x-\frac{75\theta}{1403948}\exp\Bigl(-\frac{131x}{900}\Bigr)\Biggl(\Biggr.290472+\exp\Bigl(\frac{x}{60}\Bigr)\Biggl(\Biggr.-145236}\\
&\textcolor{blue}{-343824\exp\Bigl(\frac{x}{300}\Bigr)-701974\exp\Bigl(\frac{x}{180}\Bigr)+171912\exp\Bigl(\frac{x}{50}\Bigr)+350987\exp\Bigl(\frac{x}{45}\Bigr)} \nonumber\\
&\textcolor{blue}{+863968\exp\Bigl(\frac{23x}{900}\Bigr)+443352\exp\Bigl(\frac{x}{36}\Bigr)-431984\exp\Bigl(\frac{19x}{450}\Bigr)-221676\exp\Bigl(\frac{2x}{45}\Bigr)} \nonumber\\
&\textcolor{blue}{-580944\exp\Bigl(\frac{43x}{900}\Bigr)\Biggl. \Biggr) \Biggl. \Biggr)-\frac{450}{29}\bigl(1+\theta\bigr)\exp\Bigl(-\frac{29x}{450}\Bigr),}\nonumber
\end{align*}

\begin{align*}
\textcolor{blue}{m_2(x)=}&\textcolor{blue}{150-90 \exp\Bigl(-\frac{x}{45}\Bigr)-\frac{45}{2} \Bigl(1-\exp\Bigl(-\frac{2x}{45}\Bigr)\Bigr)-\frac{360}{7} \Bigl(1-\exp\Bigl(-\frac{7x}{180}\Bigr)\Bigr)}\\
&\textcolor{blue}{+\frac{2}{1463} \exp\Bigl(-\frac{32x}{225}\Bigr) \Biggl(\Biggr.5985
\Bigl(-2+\exp\Bigl(\frac{x}{50}\Bigr)\Bigr)\theta-6930 \exp\Bigl(\frac{x}{60}\Bigr) \Bigl(-2+\exp\Bigl(\frac{x}{50}\Bigr)\Bigr) \theta}\nonumber\\
&\textcolor{blue}{-14630 \exp\Bigl(\frac{x}{45}\Bigr) \Bigl(-2+\exp\Bigl(\frac{x}{50}\Bigr)\Bigr) \theta+17556
\exp\Bigl(\frac{7x}{180}\Bigr) \Bigl(-2+\exp\Bigl(\frac{x}{50}\Bigr)\Bigr) \theta}\nonumber\\
&\textcolor{blue}{+9405 \exp\Bigl(\frac{2x}{45}\Bigr) \Bigl(-2+\exp\Bigl(\frac{x}{50}\Bigr)\Bigr) \theta-11970 \exp\Bigl(\frac{11x}{180}\Bigr) \Bigl(-2 \theta+\exp\Bigl(\frac{x}{50}\Bigr)
\Bigl(1+\theta\Bigr)\Bigr)}\nonumber\\
&\textcolor{blue}{+2 \exp\Bigl(\frac{11x}{90}\Bigr) \Bigl(-584 \theta+\exp\Bigl(\frac{x}{50}\Bigr) \Bigl(5985+292 \theta\Bigr)\Bigr)\Biggl.\Biggr)}\nonumber
\end{align*}

\noindent \textcolor{blue}{and}

\begin{align*}
\hspace{-5.0cm}\textcolor{blue}{m_3(x) =}& \textcolor{blue}{110+\frac{450}{19} \exp\Bigl(-\frac{19x}{450}\Bigr)-50 \exp\Bigl(-\frac{x}{50}\Bigr)-\frac{300}{11}\Bigl (1-\exp\Bigl(-\frac{11x}{300}\Bigr)\Bigr)}\\
&\textcolor{blue}{-\frac{900}{53} \theta\exp\Bigl(-\frac{7x}{50}\Bigr) +\frac{1800}{91}\theta
\exp\Bigl(-\frac{37x}{300}\Bigr) +\frac{225}{11}\theta \exp\Bigl(-\frac{3x}{25}\Bigr) }\nonumber\\
&\textcolor{blue}{+\frac{67050}{2279}\theta \exp\Bigl(-\frac{53x}{450}\Bigr) -\frac{1800}{73}\theta \exp\Bigl(-\frac{31x}{300}\Bigr) -\frac{227700}{6461}\theta\exp\Bigl(-\frac{91x}{900}\Bigr)}\nonumber\\
&\textcolor{blue}{-\frac{13725}{374}\theta \exp\Bigl(-\frac{22x}{225}\Bigr) -\frac{450}{43} \theta\exp\Bigl(-\frac{43x}{450}\Bigr) +\frac{179100}{3869}\theta \exp\Bigl(-\frac{73x}{900}\Bigr)
}\nonumber\\
&\textcolor{blue}{+\frac{900}{71} \theta\exp\Bigl(-\frac{71x}{900}\Bigr) +\frac{225}{17} \theta\exp\Bigl(-\frac{17x}{225}\Bigr) -\frac{900}{53} \exp\Bigl(-\frac{53x}{900}\Bigr) \Bigl(1+\theta\Bigr).}\nonumber
\end{align*}


\begin{table}
{\rowcolors{2}{gray!20!white!50}{gray!70!white!30}
\begin{center}
\begin{tabular}{ |p{1.5cm}|p{2cm}|p{2cm}|p{2cm}|p{2cm}| }
\hline
\hspace{0.5cm}$\theta$& \hspace{0.75cm}$\hat{R}_1^2$ & \hspace{0.75cm}$\hat{R}_2^2$ & \hspace{0.75cm}$\hat{R}_3^2$ &\hspace{0.75cm} $\hat{R}_4^2$\\
\hline
\hspace{0.5cm}0     & 0.8940325  & 0.02977645 & 0.005617499 & 0.005619555\\
\hspace{0.5cm}0.25  & 0.8933114  & 0.03030314 & 0.005534596 & 0.005527684\\
\hspace{0.5cm}0.5   & 0.8928696  & 0.03070126 & 0.005434534 & 0.005429036\\
\hspace{0.5cm}0.75  & 0.8928973  & 0.03102241 & 0.00532042  & 0.005313731\\
\hspace{0.5cm}1     & 0.8923792  & 0.03140672 & 0.005217848 & 0.005213205\\

\hline
\end{tabular}
\end{center}}
\caption{\textcolor{blue}{Approximations of $R_i^2$ for $i = 1,2,3,4$ in Example \ref{Example51}, depending on the dependence parameter $\theta$.}}
\label{Tab03}   
\end{table}

\textcolor{blue}{Note that $m_4(x) = m_3(x)$ for all $x\geq 0$. Figure \ref{fig_Cexmple01} (left) represents the plots of $m_i(x)$ for $i=1,2,3,4$ and $\theta = 1$. We could use these curves to estimate the expected lifetime of the system when the failure time of a component is known.}

\textcolor{blue}{To approximate the values of $R_i^2$ for $i=1,2,3,4$, we simulate 5000 lifetimes of each component with dependence parameter $\theta=1$ in \eqref{CopFGMdim4}. We calculate $\hat{R}_i^2$ for $i = 1,2,3,4$ from the simulated values of $m_i(X_i)$ and $T$. We repeat this procedure 1000 times and we represent the dispersion of the values $\hat{R}_i^2$ by the corresponding boxplots displayed in Figure \ref{fig_Cexmple01} (right). Table \ref{Tab03} represents some values of $\hat{R}_1^2$, $\hat{R}_2^2$, $\hat{R}_3^2$ and $\hat{R}_4^2$ for different values of the parameter $\theta$ in \eqref{CopFGMdim4}. As we can see in Table \ref{Tab03}, the component 1 is the most important of the system, followed by the second one, being the third and fourth components equally important. This ordering remains equal for all dependence parameters considered. In this example, the importance measure is robust against changes in the FGM copula given in \eqref{CopFGMdim4}.}


\begin{table}
{\rowcolors{2}{gray!20!white!50}{gray!70!white!30}
\begin{center}
\begin{tabular}{ |p{0.5cm}|p{0.5cm}|p{0.5cm}|p{0.5cm}|p{1.5cm}|p{1.5cm}|p{1.5cm}|p{1.5cm}| }
\hline
$\beta_1$&$\beta_2$&$\beta_3$ & $\beta_4$& \hspace{0.5cm}$\hat{R}_1^2$ & \hspace{0.5cm}$\hat{R}_2^2$ & \hspace{0.5cm}$\hat{R}_3^2$ &\hspace{0.5cm} $\hat{R}_4^2$\\
\hline
 1.5   & 1.5 & 1.5 & 1.5 &  0.734772 & 0.0953583  & 0.0112894  & 0.0112894\\
 1.7   & 1.5 & 1.5 & 1.5 &  0.687818 & 0.0994187  & 0.0121124  & 0.0121124\\
 1.5   & 1.7 & 1.5 & 1.5 &  0.760626 & 0.0973614  & 0.0102809  & 0.0102809\\    
 1.5   & 1.5 & 1.7 & 1.5 &  0.752707 & 0.0902128  & 0.0110156  & 0.0110155\\    

\hline
\end{tabular}
\end{center}}
\caption{\textcolor{blue}{Approximated values of $R_i^2$ for $i = 1,2,3,4,$ in Example \ref{Example51},  depending on the shape parameters $\beta_i$ with $i = 1,2,3,4$.}}
\label{Tab04}   
\end{table}

\textcolor{blue}{\textbf{Case II:} set $\theta = 1$ in \eqref{CopFGMdim4} and  consider components with increasing failure rate functions. In particular, the lifetimes of the components are modelled by Weibull distributions. The components' reliability functions are $\bar{F}_i(t)=\exp(-(t/\lambda_i)^{\beta_i})$ for $i = 1,2,3,4$, with scale parameters $\lambda_1=\lambda_2=\lambda_3=\lambda_4=11$. The corresponding regression curves can be computed by numerical integration. Table \ref{Tab04} shows the approximated values of $R_1^2$, $R_2^2$, $R_3^2$ and $R_4^2$ for different shape parameters $\beta_i$ with $i=1,2,3,4.$ As in Case I, the ordering among the components remains, and the first component is the most important. However, in this case, the capacity to explain the system lifetime has been respectively reduced and increased for the first and the rest of the components.}\\
\end{Exa}

\section{Conclusions}

In this article, we have proposed an importance measure based on \textcolor{blue}{the variance decomposition formula.} Specifically, our measure can be classified as a time-independent lifetime importance measure. We have shown that the proposed measure considers the system's structure, the dependence structure among the components and their corresponding lifetimes, and that is useful for the case of systems with dependent or independent components, as well as, homogeneous or heterogeneous components. Several examples of closed-form have been provided along the paper. We have also established some conditions to compare \textcolor{blue}{the importance of} two different components, when the exact values of the measure are difficult to obtain. \textcolor{blue}{We provide a procedure based on Monte Carlo methods to approximate these measures. Finally, we include an example in the context of naval engineering to illustrate the applicability of the importance index.}

\section*{Acknowledgements} 
 \textcolor{blue}{The authors acknowledge the valuable comments and suggestions made by the anonymous referees and the Associated Editor.} AA, MAS and ASL acknowledge support received by Ministerio de Econom\'ia y Competitividad of Spain under grant PID2020-116216GB-I00,  by the 2014–2020 ERDF Operational Programme and by the Department of Economy, Knowledge, Business and University of the Regional Government of Andalusia, Spain under grant: FEDER-UCA18-107519. JN thanks the partial support of Ministerio de Ciencia e Innovaci\'on of Spain under grants PID2019-103971GB-I00/AEI/10.13039/501100011033. \textcolor{blue}{The authors AA and JN state that this manuscript is part of the project TED2021-129813A-I00 and they thank the support of MCIN/AEI/10.13039/501100011033 and the European Union NextGenerationEU/PRTR.}

\end{document}